\newif\ifllncs\llncsfalse
\newif\ifanon\anonfalse

\ifllncs
\documentclass[runningheads,envcountsame]{llncs}

\else
\documentclass[11pt,letter]{article}
\fi

\newif\ifnotes
\notesfalse

\title{Parallel Spooky Pebbling Makes Regev Factoring More Practical}

\ifanon
    \author{}
    \institute{}
    \date{}
\else
    \author{
    Gregory D. Kahanamoku-Meyer\thanks{Email: \texttt{gkm@mit.edu}. Supported by U.S. DoE Co-design Center for Quantum Advantage (C$^2$QA) DE-SC0012704.}\\MIT \and
    Seyoon Ragavan\thanks{Email: \texttt{sragavan@mit.edu}. Supported by an Akamai Presidential Fellowship, NSF CNS-2154149, a Simons Investigator Award, the Defense Advanced Research Projects Agency (DARPA) under Contract No. HR0011-25-C-0300, and Amazon Research Awards. Any opinions, findings and conclusions or recommendations expressed in this material are those of the author(s) and do not necessarily reflect the views of the Defense Advanced Research Projects Agency (DARPA).
    Supported in part by Jane Street.
    This work was done in part while the author was visiting the Simons Institute for the Theory of Computing and the Challenge Institute for Quantum Computation at UC Berkeley.}\\MIT
    \and
    Katherine Van Kirk\thanks{Email: \texttt{kvankirk@g.harvard.edu}. Supported by the Fannie and John Hertz Foundation and an NDSEG fellowship.}\\Harvard
    }
    \date{}
\fi

\ifllncs\else
\date{\today} \fi

\usepackage[T1]{fontenc}
\usepackage{amssymb}
\usepackage{amsmath}
\usepackage{amsthm}
\usepackage{braket}
\PassOptionsToPackage{hyphens}{url}
\usepackage[ruled, algosection]{algorithm2e}
\usepackage{ulem}
\usepackage{multirow}
\usepackage{float}

\usepackage[utf8]{inputenc}
\usepackage{beton}
\usepackage{fullpage}
\usepackage[margin=1in]{geometry}
\usepackage{caption}
\usepackage[table,dvipsnames]{xcolor}
\definecolor{DarkBlue}{RGB}{0,0,150}
\definecolor{NotSoDarkBlue}{RGB}{15,15,210}
\definecolor{DarkRed}{RGB}{150,0,0}
\definecolor{DarkGreen}{RGB}{0,100,0}
\usepackage[pdfstartview=FitH,pdfpagemode=UseNone,colorlinks,linkcolor=DarkRed,filecolor=blue,citecolor=DarkRed,urlcolor=DarkRed,pagebackref,breaklinks]{hyperref}
\usepackage{microtype}
\usepackage{libertine}
\usepackage{libertinust1math}

\usepackage{breakcites}
\usepackage[noadjust,sort,nocompress]{cite}

\usepackage{tikz,pgfplots}
\usetikzlibrary{calc}

\usepackage{mathtools, xparse}
\usepackage{graphicx}
\usepackage{rotating}
\usepackage{url}
\usepackage{bm}
\usepackage{soul}

\usepackage{adjustbox}

\usepackage{orcidlink}
\usepackage{booktabs}

\usepackage{bbm}
\usepackage{mathtools}

\usepackage{booktabs,tabularx,makecell,adjustbox,multirow,array}
\newcolumntype{Y}{>{\centering\arraybackslash}X}

\usepackage{framed}
\usepackage{multicol}
\usepackage[capitalise]{cleveref}
\usepackage[inline]{enumitem}

\usepackage{amsthm}
\theoremstyle{plain}
\newtheorem{theorem}{Theorem}[section]{\bfseries}{\itshape}
\newtheorem{informal-theorem}[theorem]{Informal Theorem}{\bfseries}{\itshape}
{\bfseries}{}
\newtheorem{lemma}[theorem]{Lemma}{\bfseries}{\itshape}
{\bfseries}{\itshape}
\newtheorem{conjecture}{Conjecture}
\newtheorem{assumption}{Assumption}
\theoremstyle{remark}
\newtheorem{remark}[theorem]{Remark}{\itshape}{}
\theoremstyle{definition}
{\bfseries}{}
\newtheorem{definition}[theorem]{Definition}{\bfseries}{}
{\bfseries}{}
\newtheorem{corollary}[theorem]{Corollary}{\bfseries}{\itshape}
{\bfseries}{\itshape}
\newtheorem{proposition}[theorem]{Proposition}{\bfseries}{\itshape}
{\bfseries}{\itshape}

\numberwithin{theorem}{section}
\numberwithin{conjecture}{section}
\numberwithin{problem}{section}
\numberwithin{assumption}{section}

\newcommand{\RR}{\mathbb{R}}

\newcommand{\ZZ}{\mathbb{Z}}

\newcommand{\poly}{\mathsf{poly}}

\newcommand{\sbbp}{\mathsf{SBBP}}

\newcommand{\vecb}{\mathbf{b}}

\newcommand{\emp}{\mathsf{empty}}
\newcommand{\peb}{\mathsf{pebble}}
\newcommand{\gho}{\mathsf{ghost}}
\newcommand{\state}{\mathsf{state}}

\newcommand{\pebop}{\mathsf{Pebble}}
\newcommand{\unpeb}{\mathsf{Unpebble}}
\newcommand{\ghop}{\mathsf{Ghost}}

\newcommand{\blast}{\mathsf{Blast}}
\newcommand{\unblast}{\mathsf{Unblast}}

\newcommand{\istart}{i_\mathsf{start}}
\newcommand{\iend}{i_\mathsf{end}}

\allowdisplaybreaks

\DeclarePairedDelimiterX{\bigket}[1]{\bigg\lvert}{\bigg\rangle}{\,#1}
\DeclarePairedDelimiter{\norm}\lVert\rVert

\newcommand{\condparagraph}[1]{\paragraph{#1}}

\begin{document}
\maketitle

\thispagestyle{plain}
\frenchspacing

\begin{abstract}
    ``Pebble games,'' an abstraction from classical reversible computing, have found use in the design of quantum circuits for inherently sequential tasks.
    Gidney showed that allowing Hadamard basis measurements during pebble games can dramatically improve costs---an extension termed ``spooky pebble games'' because the measurements leave temporary phase errors called ghosts.
    In this work, we define and study \emph{parallel spooky pebble games}.
    Previous work by Blocki, Holman, and Lee (TCC 2022) and Gidney studied the benefits offered by either parallelism or spookiness individually; here we show that these resources can yield impressive gains when used together.
    First, we show by construction that a line graph of length $\ell$ can be pebbled in depth $2\ell$ (which is exactly optimal) using space $\leq 2.47\log \ell$.
    Then, to explore pebbling schemes using even less space, we use a highly optimized $A^*$ search implemented in Julia to find the lowest-depth parallel spooky pebbling possible for a range of concrete line graph lengths $\ell$ given a constant number of pebbles $s$.

    We show that these techniques can be applied to Regev's factoring algorithm (Journal of the ACM 2025) to significantly reduce the cost of its arithmetic.
    For example, we find that 4096-bit integers $N$ can be factored in multiplication depth 193, which outperforms the 680 required of previous variants of Regev and the 444 reported by Eker{\aa} and G{\"a}rtner for Shor's algorithm (IACR Communications in Cryptology 2025).
    While the space required for Shor's algorithm is considerably less than any variant of Regev's algorithm including ours, and thus Shor likely remains the best candidate for the first quantum factorization of large integers, our results show that implementations of Regev's algorithm are far from fully optimized, and thus Regev's algorithm may have practical importance in the future.
    We also believe our pebbling techniques will find applications in quantum cryptanalysis beyond integer factorization, and in quantum circuit compilation more broadly.
\end{abstract}

\ifllncs
\else
\thispagestyle{empty}
\newpage
\thispagestyle{empty}
\tableofcontents
\newpage
\pagenumbering{arabic}
\fi

\section{Introduction}

\emph{How can one efficiently carry out an inherently sequential computation on a quantum computer?}
It became clear early that quantum computers were unlikely to have a black-box \textit{advantage} in running inherently sequential algorithms~\cite{ozhigov_iterated_black_box_1999}; but surprisingly, due to the reversible nature of quantum computations, sequential computations can actually be \textit{more} challenging for quantum computers than classical ones.
More formally, suppose there is some fixed function $H: \{0, 1\}^n \to \{0, 1\}^n$.
Given as inputs a depth $\ell$ and an initial string $x \in \{0, 1\}^n$, we wish to evaluate $H^\ell(x) = H(H\ldots (H(x)))$.
Tasks of this form are ubiquitous --- they are used in cryptography as proofs of work or time-lock puzzles~\cite{rswtimelock,as15,at17,abp17,DBLP:conf/eurocrypt/AlwenBP18}.
They also arise when performing exponentiations in a group using repeated squaring (i.e. we have $H(x) = x^2\bmod{N}$), which appear for example in Regev's factoring algorithm~\cite{Regev23}.

While the cost of such a task is straightforwardly $\ell$ evaluations of $H$ in standard models of classical computation, it becomes significantly more interesting if we insist that all operations be reversible and $H$ is irreversible (or at least, $H$ is not efficiently reversible).
In this case, one can still just compute $H$ $\ell$ times in sequence, but this necessarily uses $O(\ell n)$ space since each iteration of $H$ must be computed out of place.
Surprisingly, one can get away with using significantly less space than this without paying prohibitive overheads in the time.
This is studied through the abstraction of \emph{pebble games}, which are played on a directed graph describing the dependencies of data in a computation (here simply a line graph of length $\ell$, as each evaluation of $H$ only requires the previous iteration's output as input).
The goal is to place a ``pebble'' on the output vertex, with the rule that pebbles can only be placed on a vertex, or removed from one, if all of that vertex's incoming edges have pebbles on their start vertices (which corresponds to all input data being available for a particular computational step).
The space complexity is captured by the maximum number of pebbles used at any one time, and the runtime is captured by the number of times pebbles are placed or removed.
For the line graph, it is known, for example, that for any $\epsilon > 0$ we can use $O(\epsilon 2^{1/\epsilon} \cdot \log \ell)$ pebbles to complete the computation in $O(\ell^{1+\epsilon})$ time steps~\cite{bennett_timespace_1989,levine_note_1990}.

An array of work has further improved on this state of affairs by adding two types of twists to pebble games:
\begin{itemize}
    \item In \emph{parallel} pebble games~\cite{DBLP:conf/tcc/BlockiHL22,DBLP:conf/eurocrypt/BlockiHL25}, pebbles may be placed and removed in parallel (as long as the parallel computations do not interfere with each other).
    The hope is that this allows one to carry out this computation in a depth-efficient manner.
    In this setting, it has been shown that one can achieve a pebbling depth of $O(\ell)$ using $O(4^{\sqrt{\log \ell}})$ pebbles for the line graph.
    This pebble count is on the one hand sub-polynomial i.e. $\ell^{o(1)}$, but on the other hand well above being poly-logarithmic.
    Moreover, it is shown that this is essentially tight in terms of space-depth volume (up to constant factors in the exponent of the number of pebbles).

    \item In the special case of quantum computation, Gidney~\cite{gidneySpookyPebble} proposed the beautiful notion of a \emph{spooky} pebble game.
    The heart of this idea is to exploit the fact that quantum computation can circumvent the need for reversibility with the use of mid-circuit measurements.
    This loophole does not come for free; the mid-circuit measurements leave a quantum phase that must be corrected later.
    However, crucially, the phase does not cost any qubits to store, reducing the space cost dramatically.

    Gidney showed that spookiness is a very powerful resource: with just $(1+o(1))\log \ell$ pebbles, we can get away with just $(1+o(1))\ell \log \ell$ pebbling steps for the line graph.\footnote{Here and throughout this paper, $\log$ will denote the base 2 logarithm.}
    Kornerup, Sadun, and Soloveichik~\cite[Theorem 3.6]{Kornerup2025tightboundsspooky} generalized this to show that with $s$ pebbles, we only need $O(m \ell)$ pebbling steps, provided that $\ell \leq \binom{m+s-2}{s-2} + 1$.
    Remarkably, they also show that this is tight up to constant factors~\cite[Theorem 4.13]{Kornerup2025tightboundsspooky}.
\end{itemize}
The main contribution of our work is to define and thoroughly study --- both asymptotically and concretely --- \emph{parallel spooky pebble games}, which capture the benefits that can be reaped by using these two resources (namely, parallelism and spookiness) in tandem.
Specifically, we show that we can achieve a pebbling depth of $2\ell$ (the best possible) using $\approx 2.47\log \ell$ pebbles,\footnote{The prefactor is equal to $1/\log \alpha$, where $\alpha \approx 1.32$ is the real solution to $\alpha^3 - \alpha - 1 = 0$. We refer the reader to Theorem~\ref{thm:parallelpebblingmain} for a formal statement.} thus asymptotically getting the best of both worlds between the low pebble count of Gidney~\cite{gidneySpookyPebble} and the low pebbling depth of~\cite{DBLP:conf/tcc/BlockiHL22}.
By the lower bound shown in~\cite[Theorem 4.13]{Kornerup2025tightboundsspooky}, this is tight up to constant factors; the number of pebbles must be $\geq \Omega(\log \ell)$ if we want a pebbling depth of $O(\ell)$.
We focus on the case of a logarithmic number of pebbles because with $O(1)$ pebbles we cannot asymptotically benefit from parallelism and logarithmically many pebbles are already sufficient for the best possible depth.
We provide an overview of these results on pebbling in Table~\ref{pebblingtable}.
\begin{table}
\begin{center}
    \begin{tabular}{|c|c|c|c|c|}
    \hline
    Works & Parallelism? & Spookiness? & Number of pebbles & Pebbling time \\
    \hline\hline
    \cite{bennett_timespace_1989,levine_note_1990} & \multirow{2}{*}{N} & \multirow{2}{*}{N} & $O(\epsilon 2^{1/\epsilon} \cdot \log \ell)$ & $O(\ell^{1+\epsilon})$ \\
    \cline{1-1}
    \cline{4-5}
    \cite{bennett_timespace_1989,DBLP:conf/tcc/BlockiHL22} & & & $O(\sqrt{\log \ell} \cdot 2^{\sqrt{\log \ell}})$ & $O(\ell \cdot 2^{\sqrt{\log \ell}})$ \\
    \hline
    \cite{DBLP:conf/tcc/BlockiHL22} & Y & N & $O(4^{\sqrt{\log \ell}})$ & $O(\ell)$ \\
    \hline
    \cite{gidneySpookyPebble,Kornerup2025tightboundsspooky} & N & Y & $\approx \log \ell$ & $\approx \ell \log \ell$ \\
    \hline
    This work & Y & Y & $\approx 2.47\log \ell$ & $2\ell$ \\
    \hline
    \end{tabular}
    \end{center}
    \caption{\label{pebblingtable}Comparison of our work and previous works on the complexity of reversibly pebbling a line graph of length $\ell$. The primary differences between these works are whether they leverage parallelism and/or spookiness as resources. Our work shows that the two resources together in tandem can asymptotically match the best of both worlds, between the space that spookiness allows and the time that parallelism allows.}
\end{table}
As an additional contribution, we provide code that via $A^\ast$ search finds the optimal-depth parallel spooky pebbling for any concrete length $\ell$ and number of pebbles $s$.\footnote{The code is available on GitHub and Zenodo: \url{https://doi.org/10.5281/zenodo.17298960}~\cite{julia_code}}
We find this numerical search to be useful in practical circumstances, where small optimizations not captured in the asymptotics can reduce the cost for specific values of $\ell$ and $s$.
Additionally, our code supports weighted cost metrics that allow the user to incorporate knowledge that certain pebbles are more space-intensive to compute than others. This is relevant for our application to factoring (as we will discuss at the end of Section~\ref{sec:techoverview}), and we are optimistic that it will be useful in other applications too.

We defer a further discussion of our results and techniques to Section~\ref{sec:techoverview}, and for now turn our attention to a specific application: quantum factoring via Regev's algorithm, where we find that parallel spooky pebble games can meaningfully improve the algorithm's implementation.
This gives us optimism that our results will find applications elsewhere in cryptography and quantum algorithms.

\paragraph{Our target application: concretely efficient Regev factoring.}
\sloppy
In the decades since Shor's $\poly(n)$ quantum algorithm for factoring $n$-bit integers was first published~\cite{shor97}, a wide array of proposals~\cite{beckman,vedral,seifert2001using,Cop02,CleveW00,beauregard,takahashi,zalka2006shors,DBLP:conf/pqcrypto/EkeraH17,gidney2017factoring,hrs17,gidney2019windowed,DBLP:journals/quantum/GidneyE21,DBLP:journals/iacr/ChevignardFS24,kahanamokumeyer2024fast,gidney2025factor2048bitrsa} have sought to ease the implementation of its algorithmic building blocks such as quantum arithmetic and the quantum Fourier transform.
In some cases, the gate count, qubit count, or depth of the factoring algorithm has been improved by constant or $\poly(\log n)$ factors through these optimizations.
However, viewed at a high level in terms of e.g. the number of $n$-bit multiplications that must be performed, the gate count of the factoring circuit remained static at $\widetilde{O}(n^2)$ for decades. We note that in these decades there were a few key innovations~\cite{seifert2001using,zalka2006shors,DBLP:conf/pqcrypto/EkeraH17,gidney2019windowed,DBLP:journals/iacr/ChevignardFS24} that give constant-factor or even log-factor improvements to the multiplication count and qubit count.

\fussy
This situation persisted until 2023, when Regev~\cite{Regev23} proposed a dramatic higher-dimensional generalization of Shor's algorithm, and showed that such a generalization --- when combined with a very careful implementation of the underlying quantum arithmetic --- provided a \emph{polynomial} reduction in the number of quantum $n$-bit multiplications required to be performed in the factoring circuits, from $O(n)$ to $O(\sqrt{n})$.
At the time, the glaring limitation of Regev's circuit was its polynomially higher space complexity compared with that of Shor: while Shor's circuit only required $O(n)$ qubits, Regev's circuit required $O(n^{3/2})$ due to its use of repeated squaring (which is not reversible) for exponentiation modulo the number $N$ to be factored.
This issue was resolved by a modification proposed by Ragavan and Vaikuntanathan~\cite{rv24} to Regev's circuit, preserving its asymptotic gate count and depth while bringing the qubit usage down to $O(n)$.

However, a closer look reveals a still unsatisfactory picture: despite its asymptotic (in fact, polynomial) improvement over Shor's circuit, even the space-efficient variant of Regev's circuit~\cite{Regev23} by~\cite{rv24} is \emph{concretely much less efficient} (when $n = 2048$, say) than state-of-the-art implementations of Shor's algorithm.
At a high level, the reason for this is the reversible Fibonacci exponentiation technique proposed by Kaliski~\cite{kal17} and leveraged by~\cite{rv24}; each step involves several multiplications of $n$-bit integers, and moreover requires significant space because it involves storing intermediate values \emph{as well as their inverses} modulo $N$. Ragavan~\cite{cryptoeprint:2024/636} proposed constant-factor improvements to the circuit by~\cite{rv24} that somewhat mitigate this gap, but the concrete efficiency remains unsatisfactory and this appears to be inherent to Fibonacci-based reversible multiplication techniques.
This gap in concrete efficiency was demonstrated in a meticulous study by Eker{\aa} and G{\"a}rtner~\cite{ekeragartnercomparison}.
The central motivating question underlying this application of our work is thus the following:
\begin{quote}
    \begin{center}
        \textit{Motivating Question: Does Regev's factoring algorithm offer a more efficient path towards quantumly factoring 2048-bit integers?}
    \end{center}
\end{quote}
We revisit the original repeated squaring proposal by Regev and propose to conserve space and depth using our parallel spooky pebbling techniques.
We additionally propose several concretely important optimizations to further cut down the costs of Regev's factoring algorithm, which we sketch in Section~\ref{sec:techoverview}.
We also implement our proposed circuits for Regev factoring in detail; in particular we report the depth of the quantum circuits (in terms of multiplication operations mod $N$, following the previous resource estimation work by~\cite{ekeragartnercomparison}).
We achieve the following guarantees:
\begin{itemize}
    \item Asymptotically, it matches the depth of~\cite{Regev23,rv24,cryptoeprint:2024/636} and is only a $O(\log n)$ factor worse in the gate and qubit counts, while enjoying much better prefactors.
    We state this formally in Theorem~\ref{thm:asymptoticwithoptimalpeb}.
    \item Concretely, for $n = 2048$, we can achieve a mod $N$ multiplication depth per run of 157, which significantly improves upon the 580 used by Fibonacci-based implementations of Regev~\cite{rv24,cryptoeprint:2024/636}. This depth also outperforms the 230 reported by~\cite{ekeragartnercomparison} for Shor's algorithm (although Shor uses much less space, see below for more nuanced comparisons).

    For $n = 4096$, we achieve a mod $N$ multiplication depth of 193, which once again significantly improves on the 680 used by Fibonacci-based implementations of Regev, and outperforms the 444 reported by~\cite{ekeragartnercomparison} for Shor.

    In order to make a fairer comparison between algorithms with variable tradeoffs between time and space, we plot the per-shot time and depth cost of our results against both~\cite{rv24} and a parallelized version of Shor (see Appendix~\ref{app:parallel_shor}) in Figure~\ref{fig:time-space-plot}.
    For more detailed results see Section~\ref{subsec:factoring-results}.
\end{itemize}

\begin{figure}
	\begin{center}
		\includegraphics[width=0.8\textwidth]{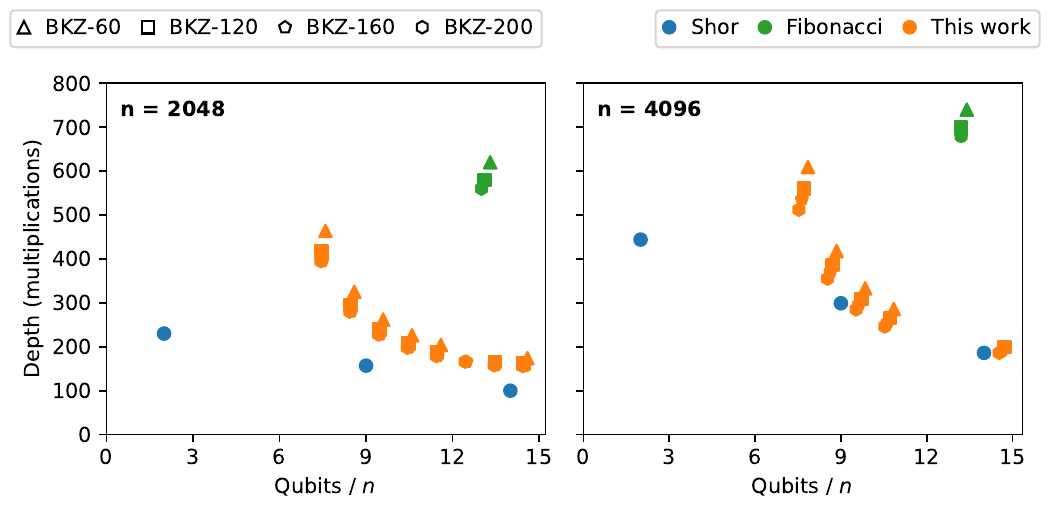}
	\end{center}
	\caption{\textbf{Depth-space tradeoffs for a single shot of various factoring algorithms.} Note that this configuration of Shor's algorithm requires roughly 10 times fewer shots than the other two constructions, which are based on Regev's algorithm. We include four points for each Regev variant corresponding to the strength of the BKZ lattice reduction used in classical postprocessing (see Section~\ref{sec:factoringcosts} for details). Across the tradeoff curve, our results outperform~\cite{rv24} (labeled ``Fibonacci'') in both space and depth, yielding by far the most efficient implementation of Regev's algorithm currently known. In some regions of parameter space the per-shot cost of our results is comparable to that of Shor's algorithm.}
	\label{fig:time-space-plot}
\end{figure}

\paragraph{Outlook and caveats.}
The practical takeaway of this work is that the cost of Regev's algorithm can be reduced considerably from the previous state of the art~\cite{Regev23,rv24,ekeragartner,ekeragartnercomparison}, even achieving a per-shot cost that is comparable to Shor's algorithm in some regimes~\cite{shor97,gidney2019windowed,DBLP:journals/quantum/GidneyE21,DBLP:journals/iacr/ChevignardFS24,gidney2025factor2048bitrsa}\footnote{Eker{\aa} and G{\"a}rtner~\cite{ekeragartnercomparison} note that they may be underestimating the benefit offered by windowing for Shor's algorithm in their cost estimates. Regardless of this, we believe the qualitative takeaways of our results --- in relation to the concrete potential of variants of Regev's algorithm --- still stand.}.
Indeed, on practical input sizes the margin between the two algorithms is narrower than previously thought.
Furthermore, we see no inherent obstacle to further optimization that builds on top of these ideas, thus reducing the cost of Regev's algorithm even further.
We thus believe that continued study on the concrete efficiency of Regev's algorithm is warranted for at least the following two reasons: (a) Shor's circuit has benefited from decades of optimizations, while Regev's circuit has only existed for two years; and (b) even if we ultimately determine that Regev does not outperform Shor on practically-sized integers, such a determination cannot be made until we are confident that we have the best practical implementation of Regev.
Our work makes substantial progress towards this goal, and we are certain that further improvements are possible.

The circuit depth per shot seems to be particularly promising, which is a metric that could become important in the future: since Regev's algorithm requires running many small circuits and classically combining the results, many intermediate-scale quantum computers could run the algorithm in parallel without any need for quantum communication between them.
This would be particularly impactful in streaming settings where it is desired to factor integers quickly as soon as they are received.

We note that in the near future, circuit depth per shot is likely not the most important metric; qubit count and the total number of gates across \emph{all runs} of the circuit are both important and by these metrics our variants of Regev are still much less efficient than Shor, as we discuss in Section~\ref{sec:pebblingresultsanalysis}.

%

\paragraph{Other applications.}
Understanding the complexity of carrying out inherently sequential computations is also integral to cryptography, where one may want to make use of proof-of-work or time-lock puzzles.
Indeed, the motivation of the study by~\cite{DBLP:conf/tcc/BlockiHL22,DBLP:conf/eurocrypt/BlockiHL25} of parallel pebbling games was to understand the post-quantum security of such tasks (while also accounting for space usage).
These works focus attention on parallelism and do not consider spookiness, due to concerns at the time with the practicality of mid-circuit measurements.
However, it has become clear that high-fidelity mid-circuit measurements will be a necessary and integral part of any future quantum computer due to their centrality in quantum error correction; indeed, high-fidelity mid-circuit measurement has already been demonstrated in a number of experiments~\cite{deist_mid-circuit_2022,koh_measurement-induced_2023,zhu_interactive_2023,hothem_measuring_2025}.

\paragraph{Organization of the paper.}
The remainder of this paper is organized as follows:
\begin{itemize}
    \item In Section~\ref{sec:techoverview}, we provide an overview of our techniques.

    \item In Section~\ref{sec:parallelspookypebbling}, we define and study abstract parallel spooky pebbling games, and prove that in this setting the line graph of length $\ell$ can be pebbled in optimal depth $2\ell$ with at most $2.47\log \ell$ pebbles.
    We also present empirical results detailing optimal pebbling depth for some concrete instances of parallel spooky pebbling games.

    \item In Section~\ref{sec:prelims}, we set up some preliminaries that will be useful for our application to Regev's factoring algorithm~\cite{Regev23}.

    \item Finally, in Section~\ref{sec:factoringcosts}, we put these ingredients together to obtain our results for reimplementing the arithmetic in Regev's factoring circuit.
\end{itemize}

\subsection{Technical Overview}\label{sec:techoverview}

In Section~\ref{sec:overviewpeb}, we provide an overview of the main ideas underlying our parallel spooky pebbling scheme.
We then turn to our factoring application in Section~\ref{sec:regevpebblingoverview}, where we discuss how Regev's circuit can be viewed as a pebbling game and additional optimizations specific to Regev factoring that we make in this work.

\subsubsection{Pebbling Games, Spooky and Parallelized}\label{sec:overviewpeb}

Recall that our goal is essentially to compute $H^{\ell}(x)$ for a function $H: \{0, 1\}^n \to \{0, 1\}^n$ and string $x \in \{0, 1\}^n$.
This can be abstracted out as a pebble game on a line graph of length $\ell$.
As already mentioned, the starting point of Gidney's spooky pebbling proposal~\cite{gidneySpookyPebble} is the following naive idea: after computing $\ket{x} \mapsto \ket{x} \ket{H(x)}$, we will be using one extra register. One way to ``clean up'' the $\ket{x}$ register would simply be to measure it in the Hadamard basis! We will obtain some classical string $s \in \left\{0, 1\right\}^n$ and be left with the state $(-1)^{\langle s, x \rangle} \ket{H(x)}$, which now does not use any additional space.
Crucially, the probability that we obtain a particular string $s$ is independent of $x$ (because all entries of the Hadamard matrix have the same magnitude), so the absolute amplitudes in our superposition will remain the same; we need only deal with these phases that arise.

Of course, this does not yet solve our problem; this phase will vary across the inputs in our superposition, and will completely destroy the correctness of the underlying algorithm if we need to compute $H^{\ell}(\cdot)$ in superposition.
We thus need to ``clean up'' this leftover phase at some point, so we will need to recompute $x$ at some point. This gives us access to three operations:
\begin{itemize}
    \item \emph{Pebbling:} given a register $\ket{x}$, write $\ket{H(x)}$ in a new register: $\ket{x} \ket{0^n} \mapsto \ket{x} \ket{H(x)}$.

    \item \emph{Ghosting:} measure a register $\ket{x}$ in the Hadamard basis. In doing so, we free up $n$ qubits but we create a phase that we will need to clean up later.

    \item \emph{Unpebbling:} this is the inverse of pebbling i.e. it computes $\ket{x} \ket{H(x)} \mapsto \ket{x} \ket{0^n}$.

    We implicitly assume that when unpebbling we will also ``exorcise'' any phases that were accumulated earlier from this stage of the pebbling game i.e. if we have a phase $(-1)^{\langle s, x \rangle}$ and also have $\ket{x}$ available in a register, we can remove this phase e.g. using phase kickback~\cite{cemm}.
\end{itemize}
This abstract formulation extends the notion of \emph{reversible pebbling games}~\cite{bennett_timespace_1989} from reversible computation to quantum measurement-based uncomputation and was first proposed and christened as ``spooky pebbling'' by Gidney~\cite{gidneySpookyPebble}. Surprisingly, these operations can be cleverly sequenced to obtain a construction that is more efficient than directly implementing the square-and-multiply algorithm. When computing modular exponents with exponents of bit length $\ell$, naive square-and-multiply uses $2\ell - 1$ multiplications (``pebbling steps'') and $\ell$ $n$-qubit registers (``pebbles''). In comparison, Gidney's proposed pebbling scheme uses $\approx \ell \log \ell$ pebbling steps (which is slightly worse), but manages to do so using only $\approx \log \ell$ pebbles. This already provides a concretely space-efficient and gate-efficient way to implement a sequential function such as $H^{\ell}$.

The main idea underlying Gidney's proposal --- which we will also make use of --- is the following divide-and-conquer approach.
Assume for now that $\ell$ is a power of 2.
The first observation is that with just two pebbles we can march one pebble all the way to position $\ell$, while leaving ghosts throughout.
We refer to this procedure as ``blasting''.
However, once we have done this, uncomputing the ghosts will be costly.
To remedy this, when blasting we will leave ``marker pebbles'' at a carefully chosen sequence of positions.
These markers will provide less costly starting points from which we can uncompute ghosts.
With this in mind, we can describe Gidney's pebbling procedure at a high level:
\begin{enumerate}
\item Blast all the way to position $\ell$, leaving marker pebbles at positions $\ell/2$, $3\ell/4$, $7\ell/8$, $\ldots$ and ghosts everywhere else.
\item\label{item:unblastrighthalf} Uncompute everything in the right half (positions $\ell/2$ to $\ell$)
\item\label{item:blastlefthalf} Recursively blast on the interval $[0, \ell/2]$.
\item Uncompute everything in the left half (positions $0$ to $\ell/2$).
\end{enumerate}
We build on Gidney's proposal by defining and analyzing \emph{parallel} spooky pebbling games, where some of the pebbling steps are run in parallel to further conserve depth.
As mentioned earlier, parallelism has previously been studied as a resource in pebble games~\cite{DBLP:conf/tcc/BlockiHL22,DBLP:conf/eurocrypt/BlockiHL25}, but without spookiness.
We show that spookiness and parallelism together give rise to pebbling games that are more efficient that what was possible using either resource individually.

One natural idea to inject parallelism into Gidney's proposal outlined above would be to run steps~\ref{item:unblastrighthalf} and~\ref{item:blastlefthalf} in parallel.
This would give a pebbling scheme with optimal depth $2\ell$ and $\approx 3\log \ell$ pebbles.
We improve on this with a very careful recursive construction based on a Fibonacci-like sequence instead of powers of 2, showing that we can achieve the same depth with just $\approx 2.47\log \ell$ pebbles (compared with the $O(4^{\sqrt{\log \ell}})$ pebbles used by~\cite{DBLP:conf/tcc/BlockiHL22} using only parallelism).
The pebbling schemes discussed in this overview are summarized in Table~\ref{pebblingtable}.

\subsubsection{Applying Pebbling Games to Regev Factoring}\label{sec:regevpebblingoverview}

\paragraph{Why pebbling?}
Let us recap the existing approaches to implementing the modular exponentiation in Regev's circuit:
\begin{itemize}
    \item Using a reversible Fibonacci exponentiation algorithm, as in Ragavan and Vaikuntanathan's optimizations~\cite{kal17,rv24,cryptoeprint:2024/636}: this incurs significant constant-factor overheads due to the need to work with intermediate registers as well as their inverses modulo $N$. This barrier seems inherent (see~\cite[``On Kaliski's Work and Quantum-Quantum Multiplication'']{rv24} for a discussion of this).

    \item Using the classic square-and-multiply algorithm, as in Regev's original circuit~\cite{Regev23}: this is now very efficient in terms of the number of multiplications, but encounters the barrier that squaring modulo $N$ is not a reversible operation and thus must be done out of place. Regev's algorithm performs $O(\sqrt{n})$ squarings mod $N$ and each squaring uses a new $n$-qubit register, which is the source of the significant $O(n^{3/2})$ qubit requirement.\footnote{We note that Shor's algorithm~\cite{shor97} circumvents this space barrier by essentially offloading the squarings mod $N$ to classical preprocessing and quantumly only doing multiplications mod $N$. This trick does not appear applicable to Regev; see~\cite[``Space Complexity'']{rv24} for a discussion of this.}
    This reversibility barrier is not as inherent; indeed, pebbling games provide a way to circumvent this.
\end{itemize}
Thus, any approach that could hope to rival the efficiency of Shor's algorithm on 2048-bit integers will likely have to modify the square-and-multiply algorithm to be more space-efficient.

\paragraph{Formulating Regev's circuit as a pebbling game.} The bottleneck of Regev's quantum circuit is computing the mapping $$(z_1, \ldots, z_d) \mapsto \prod_{j = 1}^d a_j^{z_j} \bmod{N},$$where the $a_j$'s are small bases and the $z_j$'s are non-negative integers $< D$. Regev's square-and-multiply procedure~\cite{Regev23} is summarized (in classical terms for now) in Algorithm~\ref{algo:basicrepeatedsquaring}.
Correctness follows from the fact that at the end of step $i^*$, we have $$x = \prod_{j = 1}^d a_j^{\sum_{i = 1}^{i^*} z_{j, i}2^{\log D - i}} \bmod{N},$$
where $z_{j,i}$ is the $i^\mathrm{th}$ bit of $z_j$ (in big-endian bit order, indexed from 1).

\begin{algorithm}
    \KwData{Indices $z_1, \ldots, z_d$}
    \KwResult{The product $\prod_{j = 1}^d a_j^{z_j} \bmod{N}$}
    \begin{enumerate}
        \item Initialize $x = 1$.
        \item\label{step:pebbleupdatebasic} For $i = 1, 2, \ldots, \log D$ in that order:
        \begin{enumerate}
            \item\label{step:squarebasic} Update $x \gets x^2 \bmod{N}$.
            \item\label{step:treebasic} Update $x \gets x \cdot \prod_{j = 1}^d a_j^{z_{j, i}} \bmod{N}$.
        \end{enumerate}
    \end{enumerate}
    \caption{Regev's square-and-multiply procedure. Note that for simplicity, we write the algorithm here in non-reversible terms.}\label{algo:basicrepeatedsquaring}
\end{algorithm}

In other words, each iteration of either Step~\ref{step:squarebasic} or Step~\ref{step:treebasic} of Algorithm~\ref{algo:basicrepeatedsquaring} constitutes a pebbling step in our pebbling game, which will have length $2\log D$.
We unpack this in more detail in Section~\ref{sec:peb2fac}.

\paragraph{Amortizing Step~\ref{step:treebasic} across pebbles.} As seen in Algorithm~\ref{algo:basicrepeatedsquaring}, each iteration of the loop in Step~\ref{step:pebbleupdatebasic} involves two pebbling steps: the squaring mod $N$ in Step~\ref{step:squarebasic} and the additional multiplicative update in Step~\ref{step:treebasic}. Neither of these costs is negligible. However, it was observed by Ragavan~\cite{cryptoeprint:2024/636} that one can actually save gates by only carrying out Step~\ref{step:treebasic} for some values $j$ (this is using a technique that bears some high-level similarity to the use of windowing in optimizing Shor's algorithm~\cite{gidney2019windowed,DBLP:journals/quantum/GidneyE21}; see~\cite{cryptoeprint:2024/636} for a further comparison of these techniques). Concretely, we will specify some window length $w$ and Step~\ref{step:pebbleupdatebasic} will have the following structure:

\begin{itemize}
    \item If $i$ is not a multiple of $w$: update $x \gets x^2 \bmod{N}$;
    \item If $i$ is a multiple of $w$: update $x \gets x^2 \cdot \prod_{j = 1}^d a_j^{z_{j,[i:i+w)}} \bmod N$.
\end{itemize}
Here, $z_{j,[i:i+w)}$ is the integer formed from bits $i$ to $i+w$ of $z_j$, in big-endian bit order.
This can be thought of as a batched version of Step~\ref{step:treebasic}.

\paragraph{Ghosting inside step~\ref{step:treebasic}.} We optimize further by using measurement based uncomputation to minimize the cost of reversibility inside of the pebble operations.
For $\pebop(i)$ steps from Step~\ref{step:treebasic} when $i$ is a multiple of $w$, we implement the overall transformation $x \gets x \cdot \prod_{j = 1}^d a_j^{z_{j,[i:i+w)}}$ by first computing $t_i=\prod_{j = 1}^d a_j^{z_{j,[i:i+w)}}$ in an ancillary register, and then computing the product of $x \cdot t_i$ into the output.
Then, instead of explicitly uncomputing $t_i$, we use measurement-based uncomputation to ``ghost'' it.
Only during the $\unpeb(i)$ phase must we explicitly uncompute $t_i$ (because unpebbling must not leave ghosts); during this stage we also fix all the phases associated with previous ghosting of $t_i$ during calls to $\pebop(i)$.

\paragraph{Computing $\prod_{j = 1}^d a_j^{z_j, [i:i+w)}$ in practice.}
This pebbling abstraction implicitly disregards the cost of actually computing $\prod_{j = 1}^d a_j^{z_j, [i:i+w)}$.
Asymptotically, if we set $d = \sqrt{n}$, this is reasonable since this product is typically much smaller than $N$~\cite{Regev23}.
However, in practice, we would like to set $d$ to be as large as possible to maximize the benefit we can get from high dimensionality:
Eker{\aa} and G{\"a}rtner~\cite{ekeragartnercomparison} and our work both select $d$ to be as large as possible while ensuring that $\prod_{j = 1}^d a_j^{2^w - 1} < 2^n$.
In this setting, it follows by construction that the bit length of $\prod_{j = 1}^d a_j^{z_j, [i:i+w)}$ is no longer negligible.

As discussed in their their study, Eker{\aa} and G{\"a}rtner~\cite{ekeragartnercomparison} chose to disregard the space and gate costs of constructing these trees, which was natural in their context because their finding was that Regev does not perform as well as Shor on cryptographically relevant problem sizes even when biasing the comparison in Regev's favor by omitting costs such as these.
Given our optimism about the algorithms introduced in our work, we choose to also carefully consider these costs.

In this respect, a key point of concern is the space overhead that arises from storing an entire product tree --- particularly in our case where we may be constructing multiple product trees in parallel.
We state some natural but ultimately limited ways to optimize this issue away, then outline our solution:
\begin{itemize}
    \item We could save a constant factor in space by uncomputing any right children of nodes in the product tree.
    \item Even better, when pebbling a particular node, we can ghost each layer of the tree once its parent layer has been computed. Thus we only ever need to store two layers of the tree.
    However, when \emph{unpebbling}, we will still need to store the entire tree (or a constant fraction) for the sake of uncomputation.
\end{itemize}
In our work, we simplify matters by abandoning the product tree, instead classically precomputing $a_j^{z_j 2^i}$ for $i = 0, 1, \ldots, w-1$, and simply multiplying the $d \cdot w$ relevant numbers in sequence.
This is comparably efficient when using schoolbook arithmetic, and additionally can essentially be done in place if one uses Gidney's recent low-space classical-quantum adder~\cite{gidney2025classicalquantumadderconstantworkspace} for the multiplications.

\paragraph{Not all pebbles are created equal!}
Finally, we note that our pebble operations corresponding to steps~\ref{step:squarebasic} and~\ref{step:treebasic} are qualitatively different in terms of efficiency.
The latter step needs to do some additional arithmetic (albeit with small integers), and more significantly needs additional space to store the intermediate value $t = \prod_{j = 1}^d a_j^{z_{j, [i:i+w)}}$.
Although our observation in the above paragraph allows us to shave down this space overhead considerably, it will still not be negligible.

To reflect this, we optimize our choice of pebbling while incorporating a weighted cost metric that captures the fact that pebbles in certain positions require more ancillary space to compute (or uncompute) than others.
We include this in our Julia script for finding efficient parallel spooky pebbling games, and are optimistic that the ability to weight different pebbling indices might be useful for other applications as well.

\section{Parallel Spooky Pebbling Games}\label{sec:parallelspookypebbling}

\subsection{Definitions}\label{sec:pebdefs}

\begin{definition}[Pebbling Games and States]\label{def:pebgamestate}
    A \emph{pebbling game} is specified by a target length $\ell$. At any point, the state of the game will be a vector $\state \in \left\{\peb, \gho, \emp\right\}^{\ell + 1}$. We index the elements of $\state$ by indices in $[0, \ell]$, and use $\state_i$ to denote entry $i$ of $\state$.

    The initial and final states of the pebbling game will both be $(\peb$, $ \emp$, $ \emp$, $ \ldots$, $ \emp)$.
    (The first entry of the state is a dummy variable which will always be $\peb$.)
    We require that one of the intermediate states of the pebbling game has a pebble in position $\ell$ i.e. $\state_\ell = \peb$,
    and the target end state will be $(\peb, \emp, \ldots, \emp)$.
\end{definition}

\begin{definition}[Pebbling Moves]\label{def:pebblingops}
    A \emph{pebbling move} is defined together with an \emph{active set}, which is a subset of $[0, \ell]$. It can be applied to an input state $\state$ (subject to certain restrictions on $\state$) and produces a modified state $\state'$. It can be any one of the following:
    \begin{itemize}
        \item $\pebop(i)$ for some $i \in [1, \ell]$: this is applicable when $\state_{i-1} = \peb$ and $\state_i \in \left\{\gho, \emp\right\}$. It will produce $\state'$ which is identical to $\state$ in all positions except $\state'_i$ which will now be $\peb$. Its active set is $\left\{i-1, i\right\}$.

        \item $\unpeb(i)$ for some $i \in [1, \ell]$: this is applicable when $\state_{i-1} = \state_i = \peb$. It will produce $\state'$ which is identical to $\state$ in all positions except $\state'_i$ which will now be $\emp$. Its active set is $\left\{i-1, i\right\}$.

        \item $\ghop(i)$ for some $i \in [1, \ell]$: this is applicable when $\state_i = \peb$. It will produce $\state'$ which is identical to $\state$ in all positions except $\state'_i$ which will now be $\gho$. Its active set is $\left\{i\right\}$.

    \end{itemize}
\end{definition}

\begin{definition}[Sequential and Parallel Pebbling Games]\label{def:parallelpeb}
    In a pebbling game, each time step consists of two phases that are applied in sequence:
    \begin{itemize}
        \item In the first phase of a \emph{sequential} pebbling game, we can make at most one call to $\pebop, \unpeb$. In the first phase of a \emph{parallel} pebbling game, we can make an arbitrary number of parallel calls to $\pebop, \unpeb$, \textbf{provided that their active sets are pairwise disjoint.}

        \item In the second phase, we can make an arbitrary number of parallel calls to $\ghop$.
    \end{itemize}
\end{definition}

\begin{definition}[Pebbling Cost and Depth]\label{def:paralleldepth}
    The \emph{pebbling cost} of a pebbling game is the total number of calls to $\pebop, \unpeb$.

    We define \emph{pebbling depth} as the total number of time steps (as defined in Definition~\ref{def:parallelpeb}).
\end{definition}

\begin{remark}
    Note that in the case of a sequential pebbling game, the pebbling cost and pebbling depth will always be equal.
\end{remark}

\begin{remark}
    We exclude calls to $\ghop$ when calculating pebbling cost, because in typical applications $\ghop$ will be significantly cheaper than the other three operations.
\end{remark}

\begin{remark}[Optimal-Depth Pebbling]\label{remark:optimaldepthpeb}
    It is straightforward to see that a parallel pebbling game of length $\ell$ must necessarily have depth $\geq 2\ell$.
    This is because the position of the rightmost pebble/ghost changes by at most $\pm 1$ in each time step.
    Initially and at the end, this position is 0.
    At some point this position must be $\ell$.
\end{remark}

\begin{definition}[Pebbling Space]\label{def:pebblingspace}
    Let the states of a pebbling game in sequence be $\state[0], \state[1], \ldots, \state[T]$, where $\state[0], \state[T]$ are the initial and target states defined in Definition~\ref{def:pebgamestate}. In other words, for each $i \in [0, T-1]$, $\state[i+1]$ is obtained from $\state[i]$ by applying one time step (as defined in Definition~\ref{def:parallelpeb}).

    For each $t = 1, \ldots, T$, we define the space usage $s_t$ of time step $t$ as the number of indices $i \in [1, \ell]$ such that at least one of $\state[t-1]_i, \state[t]_i$ is $\peb$. (Note in particular that this set of indices contains the active sets of all operations performed when going from $\state[t-1]$ to $\state[t]$; this can be verified by inspecting Definition~\ref{def:pebblingops}.) Then we define the \emph{pebbling space} of the pebbling game to be $\max_{t \in [1, T]} s_t$.
\end{definition}

\subsection{Previous Results}\label{sec:pebprev}

\paragraph{Sequential spooky pebbling.}
The study of sequential spooky pebbling was initiated by Gidney~\cite{gidneySpookyPebble} and further explored by Kornerup et al~\cite{Kornerup2025tightboundsspooky}.
We refer the reader to~\cite[Figure 1]{Kornerup2025tightboundsspooky} for a summary of previous results on this front and state some key results from there for comparison:
\begin{theorem}[Sequential Spooky Pebbling]
    There exist sequential spooky pebbling games achieving any of the following guarantees:
    \begin{itemize}
        \item Pebbling space 3 and pebbling cost $O(\ell^2)$~\cite{gidneySpookyPebble};
        \item Pebbling space $(1+o(1))\log \ell$ and pebbling cost $(1+o(1))\ell \log \ell$~\cite{gidneySpookyPebble};
        \item Pebbling space $s$ and pebbling cost $O(m\ell)$ provided $\binom{m+s-2}{s-2} \geq \ell$~\cite[Lemma 3.5]{Kornerup2025tightboundsspooky}.
    \end{itemize}

    Moreover, the third item here is asymptotically tight: for any $s, m$ such that $n > \binom{m+s-2}{s-2}$, a sequential spooky pebbling of length $\ell$ must use at least $\Omega(m\ell)$ steps~\cite[Theorem 4.13]{Kornerup2025tightboundsspooky}.
\end{theorem}

\paragraph{Parallel (non-spooky) pebbling.}
The notion of parallel reversible pebbling --- which resembles our notion of pebbling but does not allow for ghosting --- was previously studied by~\cite{DBLP:conf/tcc/BlockiHL22,DBLP:conf/eurocrypt/BlockiHL25} with the motivation of studying the post-quantum security of proof-of-work puzzles.
They show the following result:
\begin{theorem}\label{Parallel Non-Spooky Pebbling}
    There exists a parallel reversible pebbling scheme with pebbling space $4^{\sqrt{\log \ell}}$ and pebbling depth $O(\ell)$~\cite[Lemma 8]{DBLP:conf/tcc/BlockiHL22}.
    Moreover, if any parallel reversible pebbling scheme achieves pebbling space $s$ and pebbling depth $d$, we must have $sd \geq \ell \cdot \Omega(2^{(\sqrt{2}-o(1))\sqrt{\log \ell}})$~\cite[Theorem 1]{DBLP:conf/eurocrypt/BlockiHL25}.
\end{theorem}

\paragraph{Lower bounds on parallel spooky pebbling.}
The aforementioned lower bound by~\cite[Theorem 4.13]{Kornerup2025tightboundsspooky} can be readily turned into a lower bound on \emph{parallel} spooky pebbling games:
\begin{corollary}[Lower Bound on Parallel Spooky Pebbling]\label{cor:lowerbounds}
    For any $s, m$ such that $\ell > \binom{m+s-2}{s-2}$, a parallel spooky pebbling game of length $\ell$ must have depth at least $\Omega(m\ell/s)$.

    Consequently, for any real constant $c \geq 2$ there exists another constant $c' > 0$ such that at least $c' \log \ell$ pebbles are needed to achieve depth $< c\ell$.
\end{corollary}
\begin{proof}
    This is immediate from the observation that with $s$ pebbles, we can only possibly carry out $s$ operations in parallel at a time.
    Hence any parallel spooky pebbling game can be converted into a sequential spooky pebbling game where the total number of operations is at most $s \times$ the depth of the parallel game.

    To see the second statement, let $\mu > 0$ be a universal constant so that the depth must be at least exactly $\frac{m\ell}{s\mu}$.
    Let $c' > 0$ be any constant such that $\binom{c\mu s + s-2}{s-2} \leq 2^{s/c'}$ for all $s$ sufficiently large.
    Then suppose for the sake of contradiction that there exists a pebbling game with $s < c'\log \ell$ pebbles and depth $< c\ell$.
    Then let $m = c\mu s$.
    Since $s < c'\log \ell$, we have $\ell > 2^{s/c'} \geq \binom{m+s-2}{s-2}$.
    Consequently, the depth must be at least $\frac{m\ell}{s\mu} = c\ell$, which is a contradiction.
\end{proof}

\subsection{Parallel Spooky Pebbling with Optimal Depth and Logarithmic Space}\label{sec:parallelpebblingmain}

The main result of this section is the following positive result for parallel spooky pebbling.
By Corollary~\ref{cor:lowerbounds}, this is tight up to constant factors.
We focus on the case of a logarithmic number of pebbles because with $O(1)$ pebbles we cannot asymptotically benefit from parallelism and as the below theorem establishes, logarithmically many pebbles are already sufficient for the best possible depth.
\begin{theorem}\label{thm:parallelpebblingmain}
    Let $\alpha \approx 1.32$ be the real solution to $\alpha^3 - \alpha - 1 = 0$.
    Then for any $\ell \geq 7$, there exists a parallel spooky pebbling scheme with pebbling space $\left\lceil\frac{\log \ell - \log 0.94}{\log \alpha} - 3 \right\rceil \leq \left\lceil 2.47\log \ell - 2.77 \right\rceil$ and pebbling depth exactly $2\ell$ (which is optimal).
\end{theorem}
\noindent
Before we delve into the details of our pebbling scheme, let us reiterate the high-level ideas that we sketched in Section~\ref{sec:techoverview}.
Assume for now that $\ell$ is a power of 2.
The first observation is that with just two pebbles we can march one pebble all the way to position $\ell$, while leaving ghosts throughout.
We refer to this procedure as ``blasting''.
However, once we have done this, uncomputing the ghosts will be costly.
To remedy this, when blasting we will leave ``marker pebbles'' at a carefully chosen sequence of positions.
These markers will provide less costly starting points from which we can uncompute ghosts.
With this in mind, we can describe a variant of our pebbling procedure at a high level:
\begin{enumerate}
\item Blast all the way to position $\ell$, leaving marker pebbles at positions $\ell/2$, $3\ell/4$, $7\ell/8$, $\ldots$ and ghosts everywhere else.
\item Uncompute everything in the right half (positions $\ell/2$ to $\ell$), while recursively blasting on the interval $[0, \ell/2]$ in parallel.
\item Uncompute everything in the left half (positions $0$ to $\ell/2$).
\end{enumerate}
It turns out that this achieves depth $2\ell$ and pebbling space $\approx 3 \log \ell$.
With a more careful and involved variant of the above pebbling game, we can reduce the constant factor in the pebbling space.
Roughly, we will place the first marker pebble at a position around $c\ell$ for some $c < 1/2$.
Informally, this enables us to get some of the uncomputation in the right half out of the way before we need to allocate additional pebbles to blast in the left half.

\paragraph{Sequence $\left\{A_k\right\}$ and its properties.}
While the above simplified outline of our pebbling scheme worked with $\ell$ being a power of 2, it will be beneficial for us to work with a different exponentially growing sequence.
We define this sequence here and state its properties.
\begin{definition}[$A$-sequence]
    Define the sequence $A_1, A_2, \ldots$ by the base cases $A_1 = 1, A_2 = 2, A_3 = 2$, and the recurrence relation $A_k = A_{k-2} + A_{k-3}$.

    For any positive integer $n$, we let $A^{-1}(n)$ denote the largest index $j$ such that $A_j \leq n$. (Note in particular that $A^{-1}(2) = 3$.)
\end{definition}

\begin{proposition}[Properties of $\left\{A_k\right\}$]\label{prop:A_properties}
We have:
\begin{enumerate}
    \item\label{item:increasing} For all $k \geq 2$, $A_{k-1} \leq A_k$, with equality if and only if $k = 3$.
    \item\label{item:partialsum} For any $k \geq 4$, we have:
    $$A_k = 2 + \sum_{i = 1}^{\lfloor (k-2)/2 \rfloor} A_{k - 2i - 1}.$$
    \item\label{item:consecutivediffs} For any $k \geq 6$, we have $A_k = A_{k-1} + A_{k-5}$.
    \item\label{item:consecutivediffsincreasing} For any $k \geq 4$, we have $A_k - A_{k-1} \geq A_{k-1} - A_{k-2}$.
    \item\label{item:Aasymptoticgrowth} For all $k \geq 7$, we have $A_k \geq 0.94 \alpha^k$, where
    $$\alpha = \frac{(9 - \sqrt{69})^{1/3} + (9 + \sqrt{69})^{1/3}}{18^{1/3}} \approx 1.32$$
    is the real root of $\alpha^3 - \alpha - 1 = 0$.
\end{enumerate}

\begin{proof}
    See Appendix~\ref{sec:Apropertiesproof}.
\end{proof}

\end{proposition}

\paragraph{Our parallel spooky pebbling scheme.}
Our pebbling procedure is presented in Algorithm~\ref{algo:A_parallel_optdepth}.
Its two main recursive ingredients are a ``blasting'' procedure Algorithm~\ref{algo:A_blast} and an ``unblasting'' (uncomputation) procedure Algorithm~\ref{algo:A_unblast}.

\begin{algorithm}[h]
    \caption{Optimal-depth parallel pebbling scheme}\label{algo:A_parallel_optdepth}
    \KwData{Pebbling game of length $\ell = A_k$, initialized to state $(\peb, \emp,\emp, \ldots, \emp, \emp)$.}
    \KwResult{A completed pebbling game in the target end state $(\peb, \emp,\emp, \ldots, \emp, \emp)$. At some intermediate stage, we will have $\state_\ell = \peb$.}
    \begin{enumerate}
        \item Call $\blast(0, \ell-1)$.
        \item Apply $\pebop(\ell)$ followed by $\unpeb(\ell)$.
        \item Call $\unblast(0, \ell-1)$.
    \end{enumerate}
\end{algorithm}

\begin{algorithm}[h]
    \caption{$\blast(\istart, \iend)$: a blasting subroutine with ``marker pebbles'' at $A_k$}\label{algo:A_blast}
    \KwData{Indices $\istart, \iend$ such that $0 \leq \istart < \iend \leq \ell-1$, and $\iend - \istart = A_k-1$, for some index $k \geq 2$.
    Initially, we have $\state_{\istart} = \peb$ and $\state_{i} \in \{\emp,\gho\}$ for $i \in [\istart+1,\iend]$.}
    \KwResult{At the end of the algorithm, we will have ``marker'' pebbles at $A_k$-sized intervals away from the start. Formally, we will have $\state_i = \peb$ for $i = \istart, \istart + A_{k-3}, \istart + A_{k-3} + A_{k-5}, \ldots, \iend-1, \iend$. All other sites are $\gho$.}
    \begin{enumerate}
        \item Base case: if $\iend - \istart = 1$ (i.e. $k \in \left\{2, 3\right\}$), $\pebop(\iend)$ and terminate.
        \item If $\iend - \istart > 1$ and $k > 3$, we recursively proceed as follows: \label{algostep:blast}
        \begin{enumerate}
            \item $\pebop(\istart+1)$
            \item For sites $j \in [\istart+2,\istart+A_{k-3}]$: $\pebop(j)$ and then $\ghop(j-1)$
            \item Blast on the remainder of the subgame by recursively calling $\blast(\istart + A_{k-3}, \iend)$.
        \end{enumerate}
    \end{enumerate}
\end{algorithm}

\begin{algorithm}
    \caption{$\unblast(\istart, \iend)$: an unblasting subroutine}\label{algo:A_unblast}
    \KwData{Indices $\istart, \iend$ such that $0 \leq \istart < \iend \leq \ell-1$ and $\iend - \istart = A_k-1$ for some index $k \geq 2$.
    Initially we have $\state_{\istart}, \state_{\istart+1}, \ldots, \state_{\iend}$ as they would be at the end of $\blast(\istart, \iend)$ (see the specification of Algorithm~\ref{algo:A_blast} for details).}
    \KwResult{At the end of the algorithm, we will have $\state_{\istart} = \peb$ and $\state_i = \emp$ for all $i$ such that $\istart < i \leq \iend$.}
    \begin{enumerate}
            \item\label{algostep:find_fib} Base case: if $\iend - \istart = 1$ (i.e. $k \in \left\{2, 3\right\}$), run $\unpeb(\iend)$ and terminate.
            \item If $\iend - \istart > 1$ and $k > 3$, we recursively proceed as follows. Except where stated, all steps are sequential.
            \label{algostep:unbla_parallelization}
            \begin{enumerate}
                \item Run the first $A_{k-2} - A_{k-3}$ steps of $\unblast(\istart + A_{k-3}, \iend)$.
                \item In parallel, run the remaining $A_{k-3} - 1$ steps of $\unblast(\istart + A_{k-3}, \iend)$   \\ as well as a $\blast(\istart, \istart + A_{k-3} - 1)$.
                \item Call $\unpeb(\istart + A_{k-3})$.
                \item Call $\unblast(\istart, \istart + A_{k-3} - 1)$.
            \end{enumerate}
        \end{enumerate}
\end{algorithm}

\paragraph{Analysis of $\blast$ and $\unblast$.}\label{sec:blastunblastanalysis}
It is fairly straightforward to see from our algorithms that they are valid pebbling games and achieve the claimed input/output behavior.
We now turn our attention to establishing the desired efficiency guarantees, initially focusing on Algorithms~\ref{algo:A_blast} and~\ref{algo:A_unblast}.
We will address their depth in Lemmas~\ref{lem:blast_algo_depth} and~\ref{lem:unblast_algo_depth}, and devote the remainder of this subsection (Definition~\ref{def:bkuk} to Lemma~\ref{lem:ukfinalbound}) to analyzing their space.

\begin{lemma}
\label{lem:blast_algo_depth}
    For all $k$ and $\istart, \iend$ such that $\iend - \istart = A_k - 1$, the algorithm $\blast(\istart, \iend)$ takes exactly $A_k - 1$ time steps.
\end{lemma}
\begin{proof}
    By strong induction on $k$, see Appendix~\ref{sec:blast_algo_depth_proof}.
\end{proof}

\begin{lemma}\label{lem:unblast_algo_depth}
    For all $k$ and $\istart, \iend$ such that $\iend - \istart = A_k - 1$, the algorithm $\unblast(\istart, \iend)$ takes exactly $A_k - 1$ pebbling time steps.
\end{lemma}
\begin{proof}
    By strong induction on $k$, see Appendix~\ref{sec:unblast_algo_depth_proof}.
\end{proof}
In the below lemmas, we will account for the space at a specific time step $t$, as formally defined in Definition~\ref{def:pebblingspace}.
To do this, we formally define the following quantities:
\begin{definition}[$b_k(t)$ and $u_k(t)$]\label{def:bkuk}
    Fix a positive integer $k \geq 2$ and consider a pebbling sub-game on a segment $[\istart, \iend]$ such that $\iend - \istart = A_k - 1$.
    \begin{itemize}
        \item During the blasting phase (Algorithm~\ref{algo:A_blast}), we label by time step $t \in [1, A_k-1]$ the step in which $\istart + t$ is pebbled.
        We let $b_k(t)$ denote the space usage of time step $t$, excluding the given pebble at $\istart$.

        \item During the unblasting phase (Algorithm~\ref{algo:A_unblast}), we label by time step $t \in [1, A_k-1]$ the step in which $\iend - t + 1$ is unpebbled.
        We let $u_k(t)$ denote the space usage of time step $t$, excluding the given pebble at $\istart$.
    \end{itemize}
\end{definition}

\begin{lemma}
\label{lem:blast_space}
    For any $k \geq 2$ and $t \in [1, A_k - 1]$, we have:
    \begin{equation} \label{eqn:time_dep_blasting_space}
        b_k(t) =  \left\lceil \frac{k  - A^{-1}(A_k - t + 1)}{2} \right\rceil  + 1 \text{   pebbles}.
    \end{equation}
\end{lemma}

\begin{proof}
    By induction on $k$, see Appendix~\ref{sec:blast_space_proof}.
\end{proof}

\begin{lemma}\label{lemma:unblast_recurrence}
    For any integer $k \geq 5$, we have:
    \begin{align*}
    u_k(t) = \begin{cases}
    1 + u_{k-2}(t),
      & \text{if } 1 \leq t \leq A_{k-2} - A_{k-3}, \\
    b_{k-3}(t - (A_{k-2} - A_{k-3})) + 1 + u_{k-2}(t),
      & \text{if } A_{k-2} - A_{k-3} + 1 \leq t \\
      & \quad \leq A_{k-2} - 1, \\
    \lfloor (k-1)/2 \rfloor,
      & \text{if } t = A_{k-2}, \\
    u_{k-3}(t - A_{k-2}),
      & \text{if } A_{k-2} + 1 \leq t \leq A_k - 1.
    \end{cases}
    \end{align*}
\end{lemma}

\begin{proof}
    This follows from the recursive formulation of Algorithm~\ref{algo:A_unblast}, see Appendix~\ref{sec:unblast_recurrence_proof}.
\end{proof}

\begin{lemma}\label{lem:ukfinalbound}
    For any $k \geq 6$, we have:
    \[
u_k(t) \leq
\begin{cases}
k - 3, & \text{if } 1 \leq t \leq A_k - A_{k-1} + 1, \\
k - 4, & \text{if } A_k - A_{k-1} + 2 \leq t \leq A_k - A_{k-2} + 1, \\
k - 5, & \text{if } A_k - A_{k-2} + 2 \leq t \leq A_k - A_{k-3} + 1, \\
k - 6, & \text{if } A_k - A_{k-3} + 2 \leq t \leq A_k - A_{k-5} + 1, \\
k - r - 5, & \text{if } A_k - A_{k-2r+1} + 2 \leq t \\
           & \quad \leq A_k - A_{k-2r-1} + 1,
           \quad \text{for any } r \in [3, \lfloor (k-3)/2 \rfloor].
\end{cases}
\]

\end{lemma}
\begin{proof}
    By strong induction on $k$, making use of the recurrence relation presented in Lemma~\ref{lemma:unblast_recurrence}.
    See Appendix~\ref{sec:ukfinalbound_proof} for details.
\end{proof}

\begin{proof}[Proof of Theorem~\ref{thm:parallelpebblingmain}]

Let $k \geq 7$ be minimal such that $A_k \geq \ell$.
It follows from Item~\ref{item:Aasymptoticgrowth} of Proposition~\ref{prop:A_properties} that $k \leq \left\lceil \frac{\log \ell - \log 0.94}{\log \alpha} \right\rceil \leq \left\lceil 2.47\log \ell + 0.23 \right\rceil$.
We will run Algorithm~\ref{algo:A_parallel_optdepth} with this value of $k$.
We associate index $i$ of the length-$A_k$ pebble game with index $i - (A_k - \ell)$ of the length-$\ell$ pebble game, dropping any operations for which $i < (A_k - \ell)$.

Because the first and last $A_k - \ell$ operations of the length-$A_k$ pebble game operate only on sites with $i < (A_k - \ell)$, the depth of the resulting length-$\ell$ pebble game is at most $2A_k - 2(A_k - \ell) = 2\ell$ (and therefore it must be exactly $2\ell$).
Thus we achieve optimal depth for any $\ell$.

To analyze the space, we look at each phase separately:
\begin{itemize}
    \item During blasting, Lemma~\ref{lem:blast_space} implies that the number of pebbles is always at most $\lceil (k-3)/2 \rceil + 1 = \lceil (k-1)/2 \rceil < k-3$.
    \item When pebbling and unpebbling site $\ell$, the number of pebbles is hence at most $\lceil (k+1)/2 \rceil \leq k-3$.
    \item When unblasting, Lemma~\ref{lem:ukfinalbound} implies that the number of pebbles is at most $k-3$.
\end{itemize}
Thus the pebbling space is $\leq k-3 \leq \left\lceil 2.47\log \ell - 2.77 \right\rceil$, as claimed.
\end{proof}

\subsection{Exactly Optimal Pebbling Schemes via A* Search}\label{sec:julia}

In Section~\ref{sec:parallelspookypebbling} we construct a parallel pebbling scheme that achieves the exactly optimal depth of $2\ell$ using at most $\left\lceil 2.47\log \ell - 2.77 \right\rceil$ pebbles.
However, there are several reasons that this is not the end of the story, in particular for concrete, practical use of pebble games.
These include:
\begin{itemize}
    \item It is not clear from established lower bounds that the constant factor on this pebble count is tight;
    \item One may desire to hold the number of pebbles $s$ constant, and find the optimal-depth solution under that constraint (which may be greater than $2\ell$);
    \item There may be many solutions achieving that same optimal depth for a given $s$, and it may be desirable to use some other metric (e.g. minimizing total number of pebble operations) to choose among them; 
    \item For practical situations (see Section~\ref{sec:factoringcosts}), some extra ancillary space may be needed inside of $\pebop(i)$ and $\unpeb(i)$, and the amount may even vary with $i$.
\end{itemize}
For these reasons, we have implemented in the Julia programming language a highly optimized A* search to find the absolutely optimal pebble game for concrete values of the parameters, such as the number of pebbles $s$ and the pebble game length $\ell$.
In the left plot of Figure~\ref{fig:julia-results}, we compare the results output by this search to Theorem~\ref{thm:parallelpebblingmain}.
We find that while the algorithm of Theorem~\ref{thm:parallelpebblingmain} is nearly optimal, the constant factor still has room to be improved slightly.
In the right plot of Figure~\ref{fig:julia-results}, we plot the optimal depth achievable for constant values of $s$, as a function of $\ell$.
We see that once it is no longer possible to achieve the optimal depth of $2\ell$, the optimal depth for very small numbers of pebbles increases dramatically; meanwhile, for slightly larger numbers of pebbles, the depth remains close to optimal even as the pebble game length $\ell$ continues to increase.
The Julia code is available on GitHub and Zenodo: \url{https://doi.org/10.5281/zenodo.17298960}~\cite{julia_code}.

\begin{figure}
    \centering
    \includegraphics[width=0.49\linewidth]{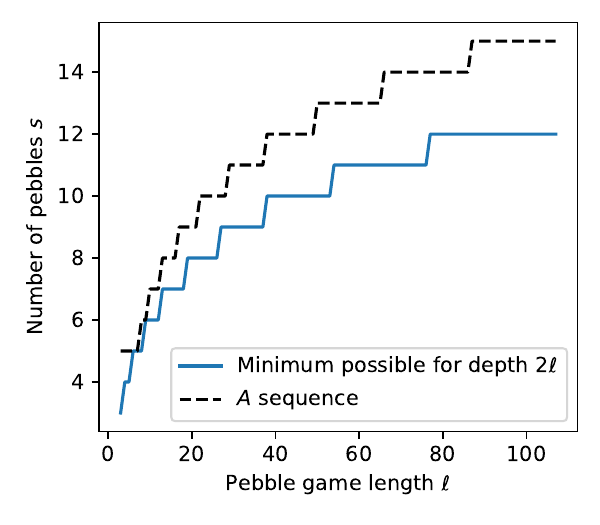}
    \includegraphics[width=0.49\linewidth]{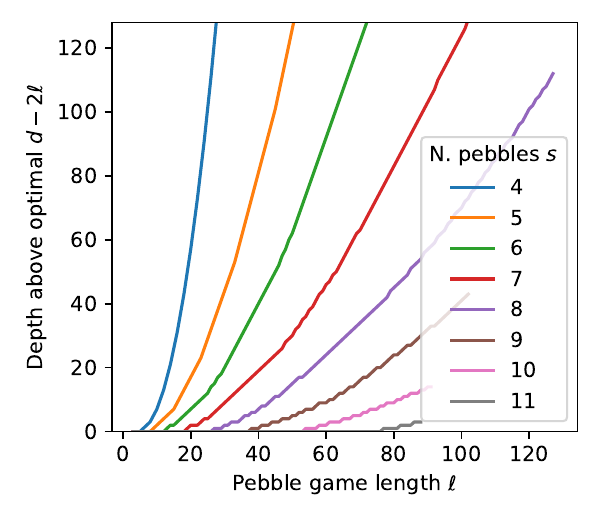}
    \caption{Results of the $A^*$ search for optimal parallel pebble games. In the left figure, we plot the minimum number of pebbles $s$ with which it is possible to achieve the absolute optimal depth $2\ell$, and compare with the space used by the $A$-sequence construction of Theorem~\ref{thm:parallelpebblingmain}. In the right figure, we plot the depth as a function of $\ell$, for different constant values of the number of pebbles $s$.}
    \label{fig:julia-results}
\end{figure}

\section{Preliminaries for Our Factoring Application}\label{sec:prelims}

Let $N < 2^n$ be an $n$-bit number that we wish to factor.
We use $\norm{\cdot}$ to denote the $\ell_2$ norm of a vector with real entries.

\subsection{Classical Algorithms for Lattice Problems}\label{sec:classicallatticealgos}

The core of the classical post-processing procedure in Regev's quantum factoring algorithm~\cite{Regev23} is an approximate solution to the following lattice problem, which we refer to as the \emph{short bounded basis problem} ($\sbbp$):

\begin{definition}[$\alpha$-approximate Short Bounded Basis Problem]
    The \emph{$\alpha$- approximate short bounded basis problem} $\sbbp_\alpha$ is the following task: given as input a basis for a lattice $\mathcal{L} \subset \RR^r$ together with a real number $T > 0$, output a list of vectors $z_1, \ldots, z_t \in \mathcal{L}$ (with $t \leq r$) such that both of the following conditions are true:
    \begin{enumerate}
        \item (Shortness) For all $i \in [t]$, we have $\norm{z_t} \leq \alpha T$.
        \item (Correctness) Any vector in $\mathcal{L}$ of norm $\leq T$ must be an integer linear combination of $z_1, \ldots, z_t$.
    \end{enumerate}
\end{definition}

\begin{remark}[Related lattice problems]
    This problem is closely related to the Short Basis Problem~\cite{DBLP:conf/icalp/Ajtai99}; in fact, they will be equivalent when $T > \lambda_r(\mathcal{L})$.
    While it is likely that $\sbbp_\alpha$ can generically be reduced to a more standard lattice problem such as $\mathsf{SVP}_{\alpha'}$ ($\alpha'$-approximate shortest vector), these reductions typically incur concrete losses in the parameters.
    Our focus is on algorithms that directly solve this problem via lattice basis reduction, following the template laid out by~\cite[Claim 5.1]{Regev23}.
\end{remark}

\paragraph{Heuristic algorithms.}
To aid our concrete resource estimates in Section~\ref{sec:factoringcosts}, we will focus on algorithms that work well in practice, rather than those that come with provable guarantees such as LLL~\cite{lenstra1982factoring} and slide reduction~\cite{gn08}.
One such algorithm is BKZ basis reduction which is analyzed heuristically using the \emph{root-Hermite factor} as a metric:

\begin{definition}[Root-Hermite factor]
    Suppose that a lattice basis reduction algorithm takes as input an arbitrary basis for a lattice $\mathcal{L} \subseteq \RR^r$ and outputs a reduced basis $\vecb_1, \ldots, \vecb_r$.
    The \emph{root-Hermite factor} achieved by this algorithm is the value $\delta_0 > 0$ such that:
    $$\norm{\vecb_1} = \delta_0^r (\det \mathcal{L})^{1/r}.$$
\end{definition}
We will work with the following heuristic, which we justify in Appendix~\ref{sec:assumptionjustification}.

\begin{assumption}\label{assumption:roothermitetosbbp}
    If a lattice basis reduction algorithm achieves root-Hermite factor $\delta_0$ in dimension $r$, then it can replace the use of LLL reduction in Regev's algorithm and yield an algorithm for $\sbbp_\alpha$, where:
    $$\alpha = \frac{\sqrt{r}}{2} \cdot \delta_0^{2r}.$$
\end{assumption}
\noindent
We will base our analyses on BKZ lattice reduction with block size $\beta = 60$, $120$, $160$, $200$.
Looking ahead, we will be applying these algorithms to lattices of dimension in the 300-600 range.
Given this, we tabulate estimates of the corresponding root-Hermite factors and costs in 600 dimensions, obtained using the online estimator\footnote{Available at \texttt{https://github.com/malb/lattice-estimator}.} by Albrecht et al~\cite{DBLP:journals/jmc/AlbrechtPS15}, in Table~\ref{roothermitetable}.
These cost estimates indicate that our setups range from being feasible for most academic institutions ($\beta = 60$) to being beyond all known implementations but still not totally inconceivable ($\beta = 200$).

\subsection{Applying $\sbbp_\alpha$ to Regev Factoring}

We refer the reader to~\cite[Section 3]{rv24} for a high-level overview of Regev's factoring algorithm, here conveying only the core points.
The broad structure of the algorithm consists of running a particular quantum circuit several times, and performing classical postprocessing to combine the measurement results in a way that reveals the factors of the input $N$.
The main operation in the quantum circuit is the application of following unitary (here defined by its action on basis states):
\begin{equation}
    \mathcal{U}(\mathbf{a}, N) \ket{\mathbf{z}}\ket{0} = \ket{\mathbf{z}}\ket{\prod_{j} a_j^{z_j} \bmod N}
    \label{eq:multiexp-unitary}
\end{equation}
where $\mathbf{z} \in \mathbb{Z}^n$ and $\mathbf{a} \in \mathbb{Z}^n_N$ is classical.
It is the implementation of this unitary to which we apply our pebbling techniques.

With this in mind, let us begin with some notation.
\begin{itemize}
    \item Let $d$ the dimension of the lattice considered by Regev~\cite{Regev23}, which the algorithm designer is free to select;
    \item let $R$ denote the radius of the initial quantum superposition over $\ZZ^d$;
    \item let $m \geq d+4$ be the number of times we run and measure the quantum circuit; and
    \item let $r = d+m$.
\end{itemize}
For clarity, we also explicitly restate Regev's number-theoretic assumption in the setting of general dimensions $d$.
We work with the variant used by Eker{\aa} and G{\"a}rtner~\cite[Assumption 1]{ekeragartner}, which is stronger but comes with the (relatively minor) benefit of allowing one to work with primes rather than squares of primes as the bases.

\begin{conjecture}[Assumption 1 of~\cite{ekeragartner}]\label{conjecture:regevnt}
    For the base integers $a_1, \ldots, a_d$, define the following lattice in $d$ dimensions:
    \begin{align*}
        \mathcal{L} &= \left\{(z_1, \ldots, z_d) \in \ZZ^d: \prod_{j = 1}^d a_j^{z_j} \equiv 1 \bmod{N}\right\}.
    \end{align*}
    Then we conjecture that there exists a basis of $\mathcal{L}$ with all vectors of $\ell_2$ norm at most $2^{Cn/d}$ for some given constant $C > 0$.
\end{conjecture}
\noindent
A simple heuristic argument based on the pigeonhole principle and put forth by Regev~\cite{Regev23} suggests that taking $C = 1+\epsilon$ likely suffices.
Below, we abstract out the classical post-processing procedure of Regev~\cite{Regev23} in terms of $\sbbp_\alpha$.

\begin{lemma}[Follows from analysis in Section 5 of~\cite{Regev23}]\label{lemma:regevlattice}
    Let $\alpha > 0$ be such that the following inequality is satisfied:
    \begin{equation}\label{eqn:dimensionsinequalityprelims}
        \alpha \cdot \sqrt{m+1} \cdot 2^{Cn/d} < \frac{\sqrt{2} \cdot R}{\sqrt{d}} \cdot \left(4 \cdot 2^n\right)^{-1/m}/6.
    \end{equation}
    Then, assuming Conjecture~\ref{conjecture:regevnt} and given the output from $m$ runs of the quantum circuit as input, it is possible to classically output a nontrivial factor of $N$ with one call to an oracle that solves $\sbbp_\alpha$ in $r = d+m$ dimensions, and some additional polynomial-time operations.
\end{lemma}

\begin{remark}
    As shown in~\cite[Appendix A]{ekeragartner}, assuming a stronger variant of Conjecture~\ref{conjecture:regevnt} and making some low-cost tweaks to Regev's algorithm allows one to find the multiplicative order of an arbitrary $g \in \ZZ_N^*$.
    By~\cite{DBLP:journals/qip/Ekera21,DBLP:journals/tqc/Ekera24}, this then typically suffices to obtain a \emph{complete} factorization of $N$ rather than just one nontrivial factor.
\end{remark}

\section{Application: Concrete Resource Estimates for Regev Factoring}\label{sec:factoringcosts}

Recall from our overview in Section~\ref{sec:regevpebblingoverview} that we have \textit{asymptotically} efficient circuits implementing Regev's factoring algorithm~\cite{Regev23,rv24,cryptoeprint:2024/636}; however, for cryptographically relevant problem sizes these leave much to be desired in terms of practical costs~\cite{ekeragartnercomparison}.
In this section, we apply our pebbling techniques introduced in Section~\ref{sec:parallelspookypebbling} to make progress towards remedying this problem.
We focus on $n$-bit RSA factoring for $n \in \{2048, 4096\}$.

\subsection{From Pebbling to Factoring}\label{sec:peb2fac}

Here we explicitly describe how we cast the implementation of Equation~\eqref{eq:multiexp-unitary} as a pebble game.
As in Regev's original paper, our task is to compute $(z_1, \ldots, z_d) \mapsto \prod_{j = 1}^d a_j^{z_j} \bmod N$ via repeated squaring.
Here, $a_1, \ldots, a_d$ are some classically known small bases and the $z_j$'s are in quantum superposition.
They will range from 0 to $D-1$, where $D$ is a power of 2.
Regev~\cite{Regev23} originally proposed to structure this computation as follows:

\begin{enumerate}
    \item $y_0 \leftarrow 1$
    \item for $i$ from 1 to $\log D$
    \begin{enumerate}
        \item $y_{2i-1} \leftarrow y_{2i-2}^2 \bmod N$
        \item $y_{2i} \leftarrow \left(\prod_{j = 1}^d a_{j}^{z_{j,i}}\right) y_{2i-1} \bmod{N}$
    \end{enumerate}
    \item return $y_{2 \log D}$
\end{enumerate}
where $z_{j,i}$ is the $i^\mathrm{th}$ \emph{highest-order} bit of $z_j$ (indexed from 1).
This algorithm can be interpreted as a line graph of length $2\log D$, where each of the $y_*$ are vertices and the operations $y_{2i-1} \leftarrow y_{2i-2}^2 \bmod N$ and $y_{2i} \leftarrow (\prod_{j = 1}^d a_{j}^{z_{j,i}}) y_{2i-1} \bmod{N}$ correspond to the odd- and even-indexed edges, respectively.
Thus we may use pebble games to efficiently compute the final value $y_{\log D} = \prod_j a_j^{z_j} \bmod N$.
Indeed, applying our spooky pebbling scheme from Section~\ref{sec:parallelspookypebbling} yields the following theorem:

\begin{theorem}\label{thm:asymptoticwithoptimalpeb}
There exists a quantum circuit for the unitary $\mathcal{U}$ of Equation~\eqref{eq:multiexp-unitary} which has ($n$-bit) multiplication depth $4 \log D$ and uses at most $1.3 (2n + S_\times(n) + o(n)) \log \log D$ qubits of space, where $S_\times(n)$ is the number of ancilla qubits used by an $n$-bit multiplication mod $N$ and we assume $\log \left[ \prod_j a_j \right] = o(n)$.
\end{theorem}

\begin{proof}
See Appendix~\ref{app:asymptotic_factoring_proof}.
\end{proof}
\noindent
Extending on this theoretical analysis, we also numerically compute the multiplication depth for certain concrete problem sizes (see Section~\ref{subsec:factoring-results}).
For practical efficiency in that analysis we make the following further optimizations, alluded to earlier in Section~\ref{sec:techoverview}:
\begin{enumerate}
    \item We use an amortization trick introduced by~\cite{cryptoeprint:2024/636} that bears some similarity to windowing~\cite{gidney2019windowed} to reduce the frequency with which products like $\prod_j a_{j}^{z_{j,i}}$ need to be multiplied in, reducing the overall length of the pebble game.
    \item We immediately measure the final value $y_{2\log D} = \prod_i a_i^{z_i} \bmod N$ when it is produced rather than insisting on keeping it around until the end as in Equation~\eqref{eq:multiexp-unitary}.
    This allows the qubits holding that value to be recycled.
    \item We compute each product $\prod_{j = 1}^d a_j^{z_{i, j}}$ by simply computing each partial product $\prod_{j = 1}^{j_\mathrm{max}} a_j^{z_{i, j}}$ via a series of in-place classical-quantum multiplies, rather than using a product tree as in~\cite{Regev23}. This saves a significant number of qubits in practice and does not asymptotically hurt the gate count when using schoolbook arithmetic.
    This is because we do not need to store an entire product tree with $d$ leaves.
    Indeed, using Gidney's recent low-space classical-quantum adder $m$ times~\cite{gidney2025classicalquantumadderconstantworkspace}, we can multiply an $m$-qubit quantum value by a classical one in-place using only two total ancilla qubits.
    \item In $\pebop(2i)$, rather than explicitly uncomputing the temporary value $t_{i}=\prod_{j = 1}^d a_j^{z_{i, j}}$, we use measurement-based uncomputation to ``ghost'' it.
    \sloppy
    \item We make our pebbling strategies \textit{ancilla-aware}---numerically optimizing them in a way that accounts for the fact that different pebble steps (multiplications and squarings) require different numbers of ancilla qubits.
\end{enumerate}
\fussy
We record our concrete algorithm with these optimizations by explicitly defining the operations $\pebop$ and $\unpeb$ (see Section~\ref{sec:parallelspookypebbling}) in Algorithms~\ref{algo:multiexp_pebble} and~\ref{algo:multiexp_unpebble} respectively.
In these algorithms we assume without loss of generality that the final pebble index is divisible by the window size $w$.
Also, the notation $z_{j,[i:i+w)}$ means ``the integer consisting of bits $i$ to $i+w$ of $z_j$'', in highest-bit order as usual.

Finally, we remark that our use of pebbling and spookiness allows us to avoid cumbersome overheads and complications that arise in the Fibonacci-based implementations~\cite{rv24,cryptoeprint:2024/636} of Regev's circuit (which is reversible without relying on mid-circuit measurements).
The main such overhead is the need to track the inverses mod $N$ of intermediate quantities used in the algorithm, which we completely circumvent.
Indeed, this is the main reason that our methods give concrete improvements over the Fibonacci-based approaches.

\begin{algorithm}[h]
    \caption{$\pebop(i)$ for Regev factoring}
    \label{algo:multiexp_pebble}
    \KwData{Classical values $\mathbf{a}$, $N$, window size parameter $w$, and pebble index $i \in [\log D-w+\log D/w]$. Quantum registers $\mathbf{z}$, $\mathsf{in}$, and $\mathsf{out}$; at the start, $\mathsf{in}$ contains $y_{i-1}$ and $\mathsf{out}$ contains $0$.}
    \KwResult{$\mathsf{out}$ is set to $y_i$, and the other registers are unchanged.}

    \eIf{$i \not \equiv 1 \pmod{w+1}$}
    {
        $\mathsf{out} \gets \mathsf{in}^2\bmod{N}$
    }{
        Classically compute $i' = (i-1)w/(w+1)$. \\
        $\mathsf{t} \gets \prod_j a_j^{z_{j,[i':i'+w)}}$\\
        $\mathsf{out} \gets \mathsf{t} * \mathsf{in} \bmod{N}$\\
        $\mathsf{Ghost}(\mathsf{t})$
    }
\end{algorithm}

\begin{algorithm}
    \caption{$\unpeb(i)$ for Regev factoring}
    \label{algo:multiexp_unpebble}
    \KwData{Classical values $\mathbf{a}$, $N$, window size parameter $w$, and pebble index $i \in [\log D-w+\log D/w]$. Quantum registers $\mathbf{z}$, $\mathsf{in}$, and $\mathsf{out}$; at the start, $\mathsf{in}$ contains $y_{i-1}$ and $\mathsf{out}$ contains $y_i$.}
    \KwResult{$\mathsf{out}$ is set to $0$, the other registers are unchanged, and all ghosts of intermediate values are removed.}

    \eIf{$i \not \equiv 1 \pmod{w+1}$}
    {
        $\mathsf{out} \gets \mathsf{out}-\mathsf{in}^2 \bmod{N}$
    }{
        Classically compute $i' = (i-1)w/(w+1)$. \\
        $\mathsf{t} \gets \prod_j a_j^{z_{j,[i':i'+w)}}$ \\
        $\mathsf{out} \gets \mathsf{out}-\mathsf{t} * \mathsf{in} \bmod{N}$\\
        $\mathsf{Unghost}(\mathsf{t})$\\
        $\mathsf{t} \gets \mathsf{t} - \prod_j a_j^{z_{j,[i':i'+w]}}$\\
    }
\end{algorithm}

\subsection{Setup}\label{sec:facsetup}

\paragraph{Efficiency metrics.}
We will measure the cost of a quantum circuit in terms of the following metrics:
\begin{itemize}
    \item We follow~\cite{ekeragartnercomparison} and estimate the gate count and depth in terms of calls to the atomic operation $\ket{u, v, t, 0^S} \rightarrow \ket{u, v, (t+uv)\bmod{N}, 0^S}$, which we hereby refer to as a ``mod $N$ multiplication''.
    In Algorithms~\ref{algo:multiexp_pebble} and~\ref{algo:multiexp_unpebble}, this means we charge a cost of 1 operation for each operation modifying $\mathsf{out}$; we do not account here for the gates/depth required to compute $\mathsf{t}$.

    We also tabulate the total number of gates across \emph{all runs} (rather than just one); as expected, we will see that the large number of runs in all variants of Regev make it much worse than Shor by this metric.

    \item We provide qubit counts, ignoring the small constant number of ancilla qubits used by Gidney's classical-quantum multiplier~\cite{gidney2025classicalquantumadderconstantworkspace} when computing $\mathsf{t}$ and also the ancilla qubits $S$ for multiplications.
    We believe that the latter cost will likely be very small in practice~\cite{kahanamokumeyer2024fast}.
    
    We note that these qubit counts are neglected consistently across our evaluations of Regev with Fibonacci arithmetic in Table~\ref{fibtable} and our evaluations of Regev with spooky pebbling in Tables~\ref{table2048} and~\ref{table4096}.
\end{itemize}

\paragraph{Classical post-processing power.}
As observed by Regev~\cite{Regev23}, his factoring algorithm and its variants enjoy natural tradeoffs between quantum circuit complexity and classical postprocessing.
To incorporate the effect of this into their cost analyses, Eker{\aa} and G{\"a}rtner consider three regimes: LLL (which runs in polynomial time), BKZ-200 (which is not practically feasible, but also not too much of an exaggeration), and perfect lattice reduction (which is an extreme case and not even close to practically feasible).

We zoom in to the range of practically feasible (or potentially feasible in the future) algorithms and thus focus on BKZ-$\beta$ for $\beta = 60, 120, 160, 200$.
As explained in Section~\ref{sec:classicallatticealgos}, these range from being very feasible with the computational resources typically available to an academic group, to being beyond anything that is currently feasible in practice.
We use Albrecht's estimator to estimate the root-Hermite factors achieved by these algorithms, which are tabulated in Table~\ref{roothermitetable}.
Looking ahead, we will plug these into Assumption~\ref{assumption:roothermitetosbbp} which we will in turn plug into Lemma~\ref{lemma:regevlattice}.

\paragraph{Setting parameters in our algorithms.}
The first parameter that we vary will be the ``window size'' $w$ that we introduced in Section~\ref{sec:regevpebblingoverview}.
Once we do this, we consider any dimension $d$ and number of runs $m$ subject to the following constraints:
\begin{itemize}
    \item $d$ must be small enough that
    $\prod_{i = 1}^d p_i^{2^w-1} \leq 2^n,$
    where $p_i$ denotes the $i$th prime.
    We denote $d_\mathrm{max} := d_\mathrm{max}(n, w)$ to be the maximal integer $d$ that would satisfy this constraint.
    This follows~\cite{ekeragartnercomparison} and is intended to ensure that the cost of computing the product $\mathsf{t} = \prod_j a_j^{z_{j,[i':i'+w)}}$ does not become dominant.
    \item $m \geq d+4$ (as required throughout the analysis by~\cite{Regev23}).
\end{itemize}
Recall from~\cite{Regev23} that $D$ is the smallest power of 2 greater than or equal to $2\sqrt{d} \cdot R$ i.e. the length of our pebbling game $\log D$ is $\lceil \log(2\sqrt{d} \cdot R) \rceil$.
For a particular $w$ and BKZ blocksize $\beta$, we will select $m, d$ that minimize $\log D$.
It can be derived from Equation~\eqref{eqn:dimensionsinequalityprelims} and Assumption~\ref{assumption:roothermitetosbbp} (setting $C \approx 1$ as in Conjecture~\ref{conjecture:regevnt}) that we require the following (see Appendix~\ref{sec:logDderivation} for details):
\begin{align}
\log D
  &> 2(m+d)\log \delta_0 + \frac{n}{d} + \frac{n}{m} \notag\\
  \quad &+ \tfrac{1}{2}\log(m+d) + \tfrac{1}{2}\log(m+1)
     + \log\!\frac{6d}{\sqrt{2}} + \tfrac{2}{m}
\label{eqn:logDbound}
\end{align}

\begin{remark}[Including lower-order terms when calculating $\log D$]\label{remark:coarselogD}
    We remark here that~\cite{ekeragartnercomparison}  ignores the lower-order terms in the above expression and approximates the lower bound on $\log D$ by just $2(m+d)\log \delta_0 + \frac{n}{d} + \frac{n}{m}$.
    While this is asymptotically a reasonable thing to do, it turns out that the concrete impact on the obtained value of $\log D$ is non-negligible.
    For example, if we work with $n = 2048$ and BKZ block size $\beta = 200$, and take $d = 222$ and $m = 303$ (as in Table~\ref{fibtable}), Equation~\eqref{eqn:logDbound} implies that we must take $\log D$ to be at least 44.
    However, if we were to discard the lower-order terms in the second line we would end up with a lower bound of only 26.
    Given this gap, we will deviate from~\cite{ekeragartnercomparison} by including these lower-order terms in our computation of $\log D$.
    Note that this shifts any comparison between Shor and Regev in the favor of Shor's algorithm.
\end{remark}

The final parameter that we vary given $\log D$ is the number of pebbles $s$ in our pebbling game, which roughly scales linearly with the qubit usage of our circuit.
Indeed, the qubit count will roughly be $d \log D + sn$.
We restrict attention to $s \in [5, 12]$ since each additional pebble adds $n$ qubits.

\subsection{Results}\label{subsec:factoring-results}

We make a comparison of our results to related factoring circuits~\cite{rv24,cryptoeprint:2024/636,DBLP:conf/pqcrypto/EkeraH17} by broadly following the lead of~\cite{ekeragartnercomparison}, but with some important differences.
In \cite{ekeragartnercomparison} the takeaway was that even with the optimizations of \cite{rv24}, Regev's algorithm is outperformed by Shor in both time \textit{and} space simultaneously on practical input sizes (e.g. 2048 bits). 
Thus, even by a cost metric that ignores space usage (a choice which is favorable to Regev) and only counts total number of multiplications, the authors were able to argue that Shor has better performance in practice.
Our parallel spooky pebble games improve the performance of Regev to an extent that the estimates for Shor from \cite{rv24} are no longer clearly better simultaneously in all metrics---indeed, for some settings of parameters, Regev's algorithm with parallel spooky pebbling may outperform in depth but require more space (see below).
For this reason, rather more nuance is required to make a fair comparison: how do we balance the cost of space versus depth?
Furthermore, how do we account for the fact that Regev's algorithm requires many independent runs of the circuit, while Shor's algorithm can be performed in one shot~\cite{DBLP:journals/qip/Ekera21} or by combining the results of several shots, making each shot cheaper~\cite{DBLP:conf/pqcrypto/EkeraH17}?

Rather than devise some single metric that attempts to capture this complexity, and then optimize parameters according to that metric, we instead opt to provide cost estimates corresponding to several different settings of parameters for each algorithm, which yield various tradeoffs of depth vs time. 
First we describe the estimates for previous algorithms to which we compare our work, and note some subtleties about them:

\begin{itemize}
    \item \textbf{Regev's algorithm~\cite{Regev23} with the optimizations by~\cite{rv24,cryptoeprint:2024/636}.}
    The concrete costs of these algorithms were studied extensively by~\cite{ekeragartnercomparison}, using the analysis in~\cite[Theorem 2.4]{cryptoeprint:2024/636}.
    Because of our interest in multiplication \emph{depth}, we make some modifications to this circuit to benefit from parallelism, which we then use to obtain concrete results for comparison with our pebbling results.
    We describe these modifications and tabulate these results in Appendix~\ref{app:regevfib}.

    As noted in Remark~\ref{remark:coarselogD}, we will incorporate lower-order terms in our computation of $\log D$.
    In other words, we calculate $\log D$ using Equation~\eqref{eqn:logDbound}, including the lower-order terms in the second line.

    \item \textbf{Shor's algorithm~\cite{shor97} combined with optimizations by Eker{\aa} and H{\aa}stad~\cite{DBLP:conf/pqcrypto/EkeraH17,DBLP:journals/dcc/Ekera20,DBLP:journals/corr/abs-2309-01754}.}
    These costs have already been calculated and reported by~\cite{ekeragartnercomparison}; in part, we will use the results from there without modifications, as tabulated in Table~\ref{shortable}.
    Recall for comparison that this variant of Shor's algorithm uses $\approx 2n$ qubits when factoring an $n$-bit integer (although smaller is achievable using different methods and at the expense of gate count~\cite{zalka2006shors,DBLP:journals/iacr/ChevignardFS24,gidney2025factor2048bitrsa}).
    No known variant of Regev's algorithm, including the present work, achieves such a small space footprint; for comparison, one may wonder what depth could be achieved if this algorithm were allowed to use more space.
    It is nontrivial to parallelize Shor's algorithm while keeping gate and qubit counts well-controlled; in Appendix~\ref{app:parallel_shor} we propose one way of doing so, which we use for the estimates we present.

    We note for completeness that~\cite{ekeragartnercomparison} incorporates windowing into their results with $w = 10$.
    As they note, this choice of $w$ may be underestimating the benefit offered by windowing for Shor's algorithm.
    For simplicity, we maintain this choice for our estimates of the cost of this algorithm; we believe the qualitative takeaways of our results --- in relation to the concrete potential of variants of Regev's algorithm --- stand regardless of this fact.
\end{itemize}

Our methodology for evaluating our spooky pebbling schemes is as follows: for any $n, \beta, w$, we select $d, m$ by iterating over all possible values subject to the constraints stated in ``Setting parameters in our algorithms'', and minimizing the required value of $\log D$ as dictated by Equation~\eqref{eqn:logDbound}.
Given $\log D$ and $s$, we then experimentally (see Section~\ref{sec:julia}) find the optimal-depth pebbling game with length $\ell := (w+1)\left\lfloor \frac{\log D - 1}{w} \right\rfloor + 1$ and $s$ pebbles and calculate the resulting depth and total number of mod $N$ multiplications.
Note that we pass information about the number of ancilla qubits required to the numerical optimization (see point 5 above).
We note that numerical experiments showed that $w < 1$ yielded worse performance (because this is equivalent to turning off windowing), but $w > 2$ yielded worse performance as well because increasing $w$ ends up leading to larger $\log D$, which ultimately increases the length $\ell$ of the pebble game.
We report results for $w = 2$ and several values of $s$ in Tables~\ref{table2048} and~\ref{table4096}. See also Figure~\ref{fig:time-space-plot} for a visual depiction of the results, including more values of the time-space tradeoff depending on number of pebbles $s$.

\begin{remark}[Choice of $s$]
    \label{remark:8pebbles}
    For most settings of the parameters in Tables~\ref{table2048} and~\ref{table4096}, we record results for $s \in \{5, 7, 12\}$, as this gives a representative picture of the tradeoffs between space and depth.
    $s=4$ is the only smaller value for which solutions to the pebble game exist at all; the depth and multiplication count of those solutions are very large (in the thousands), so we choose $s=5$ instead.
    On the other side, $s=12$ achieves the absolute optimal depth of $2\ell$ for most settings of the parameters, and thus there is no value in using more pebbles.
    The value $s=7$ achieves a reasonable tradeoff between the two.
    We note that for $n=4096$ and $\mathrm{BKZ}$-$\beta=60$, we report results for $s=8$ instead of $s=12$, as our computational resources were not sufficient to compute the optimal pebbling sequence for $s=12$ for those parameters.
\end{remark}

\subsubsection{Analysis}\label{sec:pebblingresultsanalysis}
We refer the reader to Tables~\ref{shortable}-\ref{table4096} for detailed results comparing these various factoring algorithms.

\paragraph{Comparing depth per shot.}
For $n = 2048$, Fibonacci-based approaches (Table~\ref{fibtable}) to the arithmetic in Regev's algorithm can only hope to attain a multiplication depth of 560-620 (depending on the BKZ block size that is used, see Table~\ref{fibtable}).
However when using spooky pebbling with a similar qubit count ($s=12$), \emph{our results requiring a multiplication depth of only 157-175} (Table~\ref{table2048}).
Using 5 or 7 pebbles yields somewhat worse multiplication depth (396-465 and 229-263, respectively), but still outperforms the Fibonacci variants and has the benefit of using much less space (than the other variants of Regev's algorithm; it still uses more space than Shor).
Shor's algorithm as analyzed by~\cite{ekeragartnercomparison} uses a multiplication depth of 230 if using windowing and multiple runs.
Via the parallelization strategy we describe in Appendix~\ref{app:parallel_shor}, this depth can be reduced to 100 when using a similar number of qubits as the $s=12$ spooky pebble game.

For $n = 4096$, Fibonacci-based approaches end up with a multiplication depth of 680-740, and here \emph{spooky pebbling with 8-12 pebbles fares better, requiring a multiplication depth of only 187-200 (for BKZ-120/160/200 using 12 pebbles) or 287 (for BKZ-60 using 8 pebbles)} (Table~\ref{table4096}).
Even with 5 pebbles, the multiplication depth is only 513-610.
Meanwhile, Shor's algorithm uses a multiplication depth of 444 as estimated by~\cite{ekeragartnercomparison}.
We note that in the intermediate-space regime roughly corresponding to an $s=7$ spooky pebble game, parallelized Shor (Appendix~\ref{app:parallel_shor}) requires a depth of 299 while our spooky pebble game requires depth 286-334 depending on the choice of BKZ block size.
Thus in this intermediate regime the per-shot depth of Shor and Regev is now comparable.

\paragraph{Comparing \# of mod $N$ multiplications per run.}
Even if one were to compare the total number of mod $N$ multiplications as in~\cite{ekeragartnercomparison} (as opposed to the depth), spooky pebbling still performs quite well; the gap between depth and total \# of mod $N$ multiplications per run is small.
This is in contrast to the Fibonacci case (Table~\ref{fibtable}) where there is a gap of a factor of 2.
For example, for $n = 2048$ when using spooky pebbling with 12 pebbles, the total number of mod $N$ multiplications per run ranges from 253-295, and for $n = 4096$ it ranges from 326-345.
In particular, \emph{the total number of operations for $n = 4096$ still outperforms the 444 achieved by Shor when the BKZ block size is 120, 160, or 200.}

\paragraph{Comparing \# of mod $N$ multiplications across all runs.}
As is evident from Tables~\ref{shortable}-\ref{table4096}, if we were to make this comparison counting total $\bmod{N}$ multiplications across \emph{all runs} rather than per one run (which seems to be the important metric for setting error correction parameters, since we need a large fraction of the runs to be error-free), spooky pebbling is about an order of magnitude worse than Shor, but still outperforms the Fibonacci-based results in Table~\ref{fibtable} by half an order of magnitude.

\paragraph{Comparing qubit counts.}
As far as space usage, Shor's algorithm remains the most efficient, requiring only $\approx 2n$ qubits when no parallelization is required (note that this holds for the variants considered in this work; recently it has been shown that reducing the qubit count to $0.5n + o(n)$ is possible~\cite{DBLP:journals/iacr/ChevignardFS24}).
Fibonacci variants of Regev sit on the other extreme, requiring $\geq 13n$ qubits in all cases reported in our work, and likely quite a bit more in practice since these costs disregard the space required to hold the products $\prod_{j = 1}^d a_j^{z_{j, i}}$.
Our variants of Regev require $\approx 7n$-$15n$ qubits, and these costs also account for the space required to hold the intermediate products.
Our results hence provide significant concrete improvements over Fibonacci-based implementations of Regev in all metrics, although they remain rather larger than those of Shor.

\paragraph{The effect of more lattice postprocessing power.}
We finally make a surprising observation that dramatic increases in the classical power required for the lattice postprocessing (corresponding to increasing $\beta$ in Tables~\ref{fibtable}-\ref{table4096}) have a somewhat marginal concrete effect on the number or depth of mod $N$ multiplications per run.
This is with the slight exception of spooky pebbling with $s = 5$, where the constrained space seems to enable more significant benefits from improved classical postprocessing power.

The previous study by~\cite{ekeragartnercomparison} restricted attention to the extremes of LLL (which is essentially BKZ with block size 2) and BKZ-200, and thus did not arrive at the finding that the performance is qualitatively similar even with intermediate block sizes that are much smaller (and more practicable) than $\beta = 200$.

\ifanon
\else
\condparagraph{Acknowledgements.} The authors would like to thank Craig Gidney, Isaac Chuang, Mikhail Lukin, Oded Regev, Peter Shor, Noah Stephens-Davidowitz, and Vinod Vaikuntanathan for insightful comments and discussions.
We also thank anonymous reviewers for helpful feedback, including the suggestion to consider parallelized variants of Shor's algorithm (Appendix~\ref{app:parallel_shor}) in our evaluation.
\fi

\ifllncs
\bibliographystyle{splncs04}
\bibliography{main}
\else
\bibliographystyle{alpha-doi}
\bibliography{main}
\fi

\appendix

\section{Deferred Proofs from Section~\ref{sec:parallelpebblingmain}}

\subsection{Proof of Proposition~\ref{prop:A_properties}}\label{sec:Apropertiesproof}

\begin{proof}[Proof of Items~\ref{item:increasing} and~\ref{item:partialsum}]
We prove all the assertions by strong induction on $k$. We first address the base cases where $k \leq 5$. We compute $A_1 = 1, A_2 = 2, A_3 = 2, A_4 = 3, A_5 = 4$. The claim about $A_{k-1} \leq A_k$ is now evident, and it remains to check the second assertion:
\begin{itemize}
    \item When $k = 4$, we have $A_k = 3$ and $\sum_{i = 1}^{\lfloor (k-2)/2 \rfloor} A_{k-2i-1} = A_{k-3} = 1$.
    \item When $k = 5$, we have $A_k = 4$ and $\sum_{i = 1}^{\lfloor (k-2)/2 \rfloor} A_{k-2i-1} = A_{k-3} = 2$.
\end{itemize}
Now for the inductive step. For any $k \geq 6$, we have:
\begin{align*}
    A_k &= A_{k-2} + A_{k-3} \\
    &> A_{k-4} + A_{k-3} \\
    &= A_{k-1}.
\end{align*}
The inequality is strict because $A_{k-2} \geq A_{k-3} \geq A_{k-4}$ and at least one of these inequalities is strict. Secondly, we have:
\begin{align*}
    A_k &= A_{k-3} + A_{k-2} \\
    &= A_{k-3} + 2 + \sum_{i = 1}^{\lfloor (k-4)/2 \rfloor} A_{k-2i-3} \quad \text{ (induction hypothesis)} \\
    &= A_{k-3} + 2 + \sum_{i = 2}^{\lfloor (k-2)/2 \rfloor} A_{k - 2i - 1} \\
    &= 2 + \sum_{i = 1}^{\lfloor (k-2)/2 \rfloor} A_{k-2i-1}.
\end{align*}
\end{proof}

\begin{proof}[Proof of Item~\ref{item:consecutivediffs}]
    We prove this by induction on $k$. The base case $k = 6$ can be checked manually: $A_6 = 5 = 4 + 1 = A_5 + A_1$. For the inductive step, consider any $k \geq 7$ and assume the assertion holds for $k-1$. Then we have that:
    \begin{align*}
        A_k - A_{k-1} &= (A_{k-2} + A_{k-3}) - A_{k-1} \\
        &= A_{k-3} - (A_{k-1} - A_{k-2}) \\
        &= A_{k-3} - A_{k-6} \\
        &= A_{k-5}.
    \end{align*}
\end{proof}

\begin{proof}[Proof of Item~\ref{item:consecutivediffsincreasing}]
    For $k \geq 7$, this is immediate from Items~\ref{item:increasing} and~\ref{item:consecutivediffs}. The cases $k = 4, 5, 6$ are straightforward to check manually.
\end{proof}

\begin{proof}[Proof of Item~\ref{item:Aasymptoticgrowth}]
    We prove this by strong induction.
    For $k = 7, 8, 9$, this can be verified manually.
    For $k \geq 10$, we have:
    \begin{align*}
        A_k &= A_{k-2} + A_{k-3} \\
        &\geq 0.94\left(\alpha^{k-2} + \alpha^{k-3}\right) \\
        &= 0.94\alpha^k,
    \end{align*}
    by definition of $\alpha$.
\end{proof}

\subsection{Proof of Lemma~\ref{lem:blast_algo_depth}}\label{sec:blast_algo_depth_proof}

This is straightforward by strong induction on $k$.
For $k = 2, 3$ we have $A_k - 1 = 1$ and indeed we only have one step.
For $k > 3$, the first two parts comprise $A_{k-3}$ time steps.
The recursive call to $\blast(\istart + A_{k-3}, \iend)$ uses $A_{k-2} - 1$ time steps by the induction hypothesis.
In total, this is $A_{k-3} + (A_{k-2} - 1) = A_k - 1$ time steps as desired.
\qed

\subsection{Proof of Lemma~\ref{lem:unblast_algo_depth}}\label{sec:unblast_algo_depth_proof}

This is again by strong induction on $k$.
For $k = 2, 3$, we have $A_k - 1 = 1$ and only one time step.
For $k > 3$, note by the induction hypothesis that $\unblast(\istart + A_{k-3}, \iend)$ uses $A_{k-2} - 1$ time steps and $\unblast(\istart, \istart + A_{k-3} - 1)$ uses $A_{k-3} - 1$ time steps.
Also, by Lemma~\ref{lem:blast_algo_depth}, $\blast(\istart, \istart + A_{k-3} - 1)$ uses $A_{k-3} - 1$ time steps.
We have four stages:
\begin{enumerate}
    \item The first $A_{k-2} - A_{k-3}$ time steps of $\unblast(\istart + A_{k-3}, \iend)$.
    \item An additional $A_{k-3} - 1 = (A_{k-2} - 1) - (A_{k-2} - A_{k-3})$ time steps go towards completing $\unblast(\istart + A_{k-3}, \iend)$. This is also the required number of time steps for $\blast(\istart, \istart + A_{k-3} - 1)$.
    \item One time step to unpebble $\istart + A_{k-3}$.
    \item The final $A_{k-3} - 1$ time steps to $\unblast(\istart, \istart + A_{k-3} - 1)$.
\end{enumerate}
In total, this is $(A_{k-2} - A_{k-3}) + (A_{k-3} - 1) + 1 + (A_{k-3} - 1) = A_k - 1$ time steps, which completes the proof.
\qed

\subsection{Proof of Lemma~\ref{lem:blast_space}}\label{sec:blast_space_proof}

We do this by induction on $k$. We first address the base cases where $k = 2, 3$. In this case, we simply have $b_k(1) = 1$. Indeed, we have:
$$\left\lceil \frac{k - A^{-1}(A_k - t+1)}{2}\right\rceil + 1 = \left\lceil \frac{k - A^{-1}(2)}{2}\right\rceil + 1 = \left\lceil \frac{k-3}{2} \right\rceil + 1 = 1.$$
For clarity, we also manually address the case where $k = 4$. In this case, we have $b_k(1) = 1$ and $b_k(2) = 2$, and indeed we have:
$$\left\lceil \frac{k - A^{-1}(A_k - 1+1)}{2}\right\rceil + 1 = \left\lceil \frac{4 - 4}{2}\right\rceil + 1 = 1 = b_4(1),\text{ and}$$
$$\left\lceil \frac{k - A^{-1}(A_k - 2+1)}{2}\right\rceil + 1 = \left\lceil \frac{k-3}{2}\right\rceil + 1 = 2 = b_4(2).$$
Now for the inductive step: assume that $k \geq 5$ and that we have shown the result for $k - 2$. We consider a few different ranges of $t$:
\begin{itemize}
    \item For $t = 1$, we have $b_k(t) = 1$, and indeed we have:
    $$\left\lceil \frac{k - A^{-1}(A_k - t+1)}{2}\right\rceil + 1 = \left\lceil \frac{k - A^{-1}(A_k)}{2}\right\rceil + 1 = 1.$$
    \item For $t \in [2, A_{k-3}]$, we have $b_k(t) = 2$. Moreover, we have $A_k - t + 1 \in [A_{k-2} + 1, A_k - 1] \Rightarrow A^{-1}(A_k - t + 1) \in \left\{k-2, k-1\right\}$, so it follows that:
    $$\left\lceil \frac{k - A^{-1}(A_k - t+1)}{2}\right\rceil + 1 = \left\lceil \frac{1\text{ or }2}{2}\right\rceil + 1 = 2.$$
    \item For $t \in [A_{k-3} + 1, A_k - 1]$, we have $b_k(t) = 1 + b_{k-2}(t - A_{k-3})$. The extra pebble comes from the pebble that has been placed at site $\istart + A_{k-3}$. Using the induction hypothesis, it follows that:
    \begin{align*}
        b_k(t) &= 1 + b_{k-2}(t - A_{k-3}) \\
        &= 2 + \left\lceil \frac{(k-2) - A^{-1}(A_{k-2} - (t - A_{k-3}) + 1)}{2} \right\rceil \\
        &= 2 + \left\lceil \frac{(k-2) - A^{-1}(A_k - t + 1)}{2} \right\rceil \\
        &= \left\lceil \frac{k - A^{-1}(A_k - t + 1)}{2} \right\rceil + 1,
    \end{align*}
    which completes the induction.
\end{itemize}
\qed

\subsection{Proof of Lemma~\ref{lemma:unblast_recurrence}}\label{sec:unblast_recurrence_proof}

This is straightforward from the recursive formulation of Algorithm~\ref{algo:A_unblast}. We address each case in sequence:
    \begin{itemize}
        \item For $t \in [1, A_{k-2} - A_{k-3}]$, the only thing we are doing is unblasting on the segment $[\istart + A_{k-3}, \iend]$. There is one extra pebble sitting idly in position $\istart + A_{k-3}$, which gives us a total pebble usage of $1 + u_{k-2}(t)$.
        \item For $t \in [A_{k-2} - A_{k-3} + 1, A_{k-2} - 1]$, we additionally commence blasting on the segment $[\istart, \istart + A_{k-3} - 1]$, which contributes an additional $b_{k-3}(t - (A_{k-2} - A_{k-3}))$ pebbles.
        \item At $t = A_{k-2}$, all we are doing is unpebbling site $\istart + A_{k-3}$. It is apparent from this that the pebbling usage of this step is $b_{k-3}(A_{k-3} - 1) + 1$.
        Thus we have:
        \begin{align*}
            u_k(t) &= b_{k-3}(A_{k-3} - 1) + 1 \\
            &= \left\lceil \frac{(k-3) - A^{-1}(A_{k-3} - (A_{k-3} - 1) + 1)}{2} \right\rceil + 2 \\
            &= \left\lceil \frac{(k-3) - A^{-1}(A_{k-3} - (A_{k-3} - 1) + 1)}{2} \right\rceil + 2 \\
            &= \left\lceil \frac{k-2}{2} \right\rceil \\
            &= \left\lfloor \frac{k-1}{2}\right\rfloor.
        \end{align*}

        \item For $t \in [A_{k-2} + 1, A_k-1]$, all we are doing is unblasting on the segment $[\istart, \istart + A_{k-3} - 1]$. It follows that the pebbling usage of this step is $u_{k-3}(t - A_{k-2})$.
    \end{itemize}
    This completes the proof of the lemma.
\qed

\subsection{Proof of Lemma~\ref{lem:ukfinalbound}}\label{sec:ukfinalbound_proof}

The cases $k = 6, 7, 8$ may be manually checked. For convenience, we list the relevant numbers below:
    \begin{itemize}
        \item $k = 6$: we have $A_k = 5$ and $u_k([1, 4]) = [3, 3, 2, 1]$.
        \item $k = 7$: we have $A_k = 7$ and $u_k([1, 6]) = [4, 4, 4, 3, 2, 1]$.
        \item $k = 8$: we have $A_k = 9$ and $u_k([1, 8]) = [4, 5, 5, 4, 3, 3, 2, 1]$.
    \end{itemize}

    Now for the inductive step: assume $k \geq 9$ and the result has been shown for $k-2$ and $k-3$. Before we begin, note for simplicity that when $k \geq 5$ and $t \in [A_{k-2} - A_{k-3} + 1, A_{k-2} - 1]$, we have that:
    \begin{align*}
        u_k(t) &= b_{k-3}(t - (A_{k-2} - A_{k-3})) + 1 + u_{k-2}(t)\\
        &= \left\lceil \frac{k-3 - A^{-1}(A_{k-3} - (t - (A_{k-2} - A_{k-3})) + 1)}{2} \right\rceil + 2 + u_{k-2}(t) \\
        &= \left\lceil \frac{k+1 - A^{-1}(A_{k-2} - t + 1)}{2} \right\rceil + u_{k-2}(t).
    \end{align*}
    We now just enumerate through various ranges of $t$ and check each case. Note that by Items~\ref{item:increasing} and~\ref{item:consecutivediffs} of Proposition~\ref{prop:A_properties}, we have that $A_{k-2} - A_{k-3} = A_{k-7} < A_{k-5} = A_{k-1} - A_k$ (since $k \geq 8$).
    \begin{itemize}
        \item $t \in [1, A_{k-2} - A_{k-3}]$: we have $u_k(t) = 1 + u_{k-2}(t) \leq 1 + ((k-2)-3) = k-4 < k-3$.
        \item $t = A_{k-2} - A_{k-3} + 1$: we have:
        \begin{align*}
            u_k(t) &= \left\lceil \frac{k+1 - A^{-1}(A_{k-2} - t + 1)}{2} \right\rceil + u_{k-2}(t) \\
            &= \left\lceil \frac{k+1 - A^{-1}(A_{k-3})}{2} \right\rceil + u_{k-2}(t) \\
            &= 2 + u_{k-2}(t) \text{ (Proposition~\ref{prop:A_properties}, Item~\ref{item:increasing})} \\
            &\leq 2 + ((k-2)-3) \text{ (induction hypothesis)} \\
            &= k-3.
        \end{align*}
        \item $t \in [A_{k-2} - A_{k-3} + 2, A_k - A_{k-1} + 1]$: note firstly that $A_k - A_{k-1} = A_{k-2} - (A_{k-1} - A_{k-3}) = A_{k-2} - A_{k-4} \Rightarrow A_k - A_{k-1} + 1 = A_{k-2} - (A_{k-4} - 1) \leq A_{k-2} - 3$ since $k \geq 9$. Secondly, we have $A_{k-2} - t + 1 \in [A_{k-4}, A_{k-3} - 1] \Rightarrow A^{-1}(A_{k-2} - t + 1) = k-4$. Thus we have:
        \begin{align*}
            u_k(t) &= \left\lceil \frac{k+1 - A^{-1}(A_{k-2} - t + 1)}{2} \right\rceil + u_{k-2}(t) \\
            &= \left\lceil \frac{k+1 - (k-4)}{2} \right\rceil + u_{k-2}(t) \\
            &= 3 + u_{k-2}(t) \\
            &\leq 3 + ((k-2)-4) \text{ (induction hypothesis)} \\
            &= k-3.
        \end{align*}

        \item $t \in [A_k - A_{k-1} + 2 = A_{k-2} - A_{k-4} + 2, A_{k-2} - A_{k-5} + 1]$: note that $A_{k-2} - A_{k-5} + 1 \leq A_{k-2} - 2$ since $k \geq 9$. Secondly, we have $A_{k-2} - t + 1 \in [A_{k-5}, A_{k-4} - 1] \Rightarrow A^{-1}(A_{k-2} - t + 1) = k-5$. Thus we have:
        \begin{align*}
            u_k(t) &= \left\lceil \frac{k+1 - (k-5)}{2} \right\rceil + u_{k-2}(t) \\
            &= 3 + u_{k-2}(t) \\
            &\leq 3 + ((k-2)-5) \text{ (induction hypothesis)} \\
            &= k - 4.
        \end{align*}
        This bound is sufficient for us since $t \leq A_{k-2} - A_{k-5} + 1 = A_{k-4} + 1 \leq A_{k-3} + 1 = A_k - A_{k-2} + 1$.

        \item $t \in [A_{k-2} - A_{k-5} + 2, A_{k-2} - A_{k-7} + 1]$: note that $A_{k-2} - A_{k-7} + 1 \leq A_{k-2} - 1$ since $k \geq 9$. Secondly, we have $A_{k-2} - t + 1 \in [A_{k-7}, A_{k-5} - 1] \Rightarrow A^{-1}(A_{k-2} - t + 1) \in \left\{k-7, k-6\right\}$. Thus we have:
        \begin{align*}
            u_k(t) &= \left\lceil \frac{k+1 - [(k-7)\text{ or }(k-6)]}{2} \right\rceil + u_{k-2}(t) \\
            &= 4 + u_{k-2}(t) \\
            &\leq 4 + ((k-2) - 6)\text{ (induction hypothesis)} \\
            &= k-4.
        \end{align*}
        This bound is sufficient for us since $t \leq A_{k-2} - A_{k-7} + 1 = A_{k-3} + 1 = A_k - A_{k-2} + 1$.

        \item $t \in [A_{k-2} - A_{(k-2) - 2r + 1} + 2, A_{k-2} - A_{(k-2) - 2r-1} + 1]$ for some $r \in [3, \lfloor (k-5)/2 \rfloor]$. Note that this case only arises if $\lfloor (k-5)/2 \rfloor \geq 3 \Leftrightarrow k \geq 11$. (For $k = 9, 10$, this case may be safely skipped over.)

        In this case, we have that:
        \begin{align*}
            A_{k-2} - A_{(k-2)-2r-1} + 1 &\leq A_{k-2} - A_{(k-2) - 2\lfloor (k-5)/2 \rfloor + 1} + 1 \\
            &= A_{k-2} - A_{2\text{ or }3} + 1 \\
            &= A_{k-2} - 1.
        \end{align*}
        Secondly, we have $A_{k-2} - t +1 \in [A_{(k-2) - 2r-1}, A_{(k-2) - 2r + 1} - 1] \Rightarrow$\\$A^{-1}(A_{k-2} - t + 1) \in \left\{(k-2)-2r-1, (k-2)-2r\right\}$. Thus we have:
        \begin{align*}
            u_k(t) &= \left\lceil \frac{k+1 - [(k-2r-3)\text{ or }(k-2r-2)]}{2} \right\rceil + u_{k-2}(t) \\
            &= r+2+u_{k-2}(t) \\
            &\leq r+2 + ((k-2)-r-5) \text{ (induction hypothesis)} \\
            &= k-5.
        \end{align*}
        This bound suffices since $t \leq A_{k-2} - 1 < A_k - A_{k-3} + 1$.

        \item $t = A_{k-2}$: in this case, we have $u_k(t) = \lfloor (k-1)/2 \rfloor$. Since $t < A_k - A_{k-3} + 1$, it suffices to check that $\lfloor (k-1)/2 \rfloor \leq k-5$, which holds for any $k \geq 8$.

        \item $t \in [A_{k-2} + 1 = A_k - A_{k-3} + 1, A_k - A_{k-5} + 1]$: we have
        \begin{align*}
            u_k(t) &= u_{k-3}(t - A_{k-2}) \\
            &= u_{k-3}(t' \in [1, A_{k-3} - A_{k-5} + 1]) \\
            &\leq (k-3) - 3 \text{ (induction hypothesis)} \\
            &= k-6,
        \end{align*}
        which is the desired bound.
        \item $t \in [A_k - A_{k-5} + 2, A_k - A_{k-7} + 1]$, which amounts to the $r = 3$ case: we have
        \begin{align*}
            u_k(t) &= u_{k-3}(t' \in [A_{k-3} - A_{k-5} + 2, A_{k-3} - A_{k-7} + 1]) \\
            &\leq (k-3) - 5 \text{ (induction hypothesis)} \\
            &= k - r - 5.
        \end{align*}
    \end{itemize}
    At this point, our proof is complete for $k = 9, 10$ since in this case we have $A_k - A_{k-7} + 1 = A_k - 1$. From here on, assume $k \geq 11$. We have two more cases:
    \begin{itemize}
        \item $t \in [A_k - A_{k-7} + 2, A_k - A_{k-9} + 1]$, which amounts to the $r = 4$ case: we have
        \begin{align*}
            u_k(t) &= u_{k-3}(t' \in [A_{k-3} - A_{k-7} + 2, A_{k-3} - A_{k-9} + 1]) \\
            &\leq (k-3) - 6 \\
            &= k - r - 5.
        \end{align*}
    \end{itemize}
    At this point, our proof is complete for $k = 11, 12$, since in this case we have $A_k - A_{k-9} + 1 = A_k - 1$. From here on, assume $k \geq 13$. We have one last case:
    \begin{itemize}
        \item $t \in [A_k - A_{k - 2r + 1} + 2, A_k - A_{k-2r-1} + 1]$ for some $r \in [5, \lfloor (k-3)/2 \rfloor]$. Now let $t' = t - A_{k-2} \in [A_{k-3} - A_{k-2r+1} + 2, A_{k-3} - A_{k-2r-1} + 1] \subseteq [A_{k-3} - A_{(k-3) - 2(r-2) + 1} + 2, A_{k-3} - A_{(k-3) - 2(r-1) -1} + 1]$.

        Thus we may take $r' \in \left\{r-2, r-1\right\}$ such that $t' \in [A_{k-3} - A_{(k-3) - 2r' + 1} + 2, A_{k-3} - A_{(k-3) - 2r'-1} + 1]$. Moreover, $r' \geq r-2 \geq 3$. Next we claim that $r' \leq \lfloor (k-6)/2 \rfloor$ (this will enable us to use the induction hypothesis for $r'$ and $k-3$). We address this in two cases:
        \begin{itemize}
            \item If $r \leq \lfloor (k-4)/2 \rfloor$, this follows from the fact that $r' \leq r-1 \leq \lfloor (k-6)/2 \rfloor$.

            \item If $\lfloor (k-4)/2 \rfloor < r \leq \lfloor (k-3)/2 \rfloor$, this forces $k$ to be odd and $r = (k-3)/2$. If we could show that we could take $r' = r - 2$, then we would be done since $r-2 \leq (k-7)/2$. This in turn follows from noting that $t' \leq A_{k-3} - 1 = A_{k-3} - A_{(k-3) - 2(r-2) - 1} + 1$.
        \end{itemize}
        It hence follows from the inductive hypothesis that:
        \begin{align*}
            u_k(t) &= u_{k-3}(t') \\
            &\leq (k-3)-r'-5 \\
            &\leq (k-3)-(r-2) - 5 \\
            &= k-r-6,
        \end{align*}
        as desired.
    \end{itemize}
    This exhausts all possible values of $t$ since $A_k - A_{k - 2\lfloor (k-3)/2 \rfloor -1} + 1 = A_k - A_{2\text{ or }3} + 1 = A_k - 1$. Thus our induction and the proof of the lemma are complete.
\qed

\section{Details of A* Search}

In Section~\ref{sec:julia}, we describe how we use an optimized A* search, implemented in the Julia programming language, to explicitly find the exact optimal parallel pebbling sequences for certain values of $\ell$ and $s$.
In this Appendix, we describe details of the design of the A* search.

We first must precisely define the problem we are solving:
\begin{center}
    \textit{
        Given a particular number of pebbles $s$ and pebble game length $\ell$, find a series of parallel pebbling steps that has the absolute minimum possible depth.
    }
\end{center}
We cast this as a graph search problem: consider a directed graph whose vertices are all possible configurations of $\leq s$ pebbles (and $\leq \ell$ ghosts) on $\ell$ indices, and whose edges are valid parallel pebbling steps (which may consist of several moves done in parallel).
The goal is to find the shortest path between the ``initial'' vertex $V_i$ (corresponding to no pebbles present) and the ``final'' vertex $V_f$ (corresponding to a single pebble at position $\ell$ and no other pebbles or ghosts present).

This graph search problem initially seems rather hopeless for all but very small $\ell$ and $s$.
There are $\sum_{s'=0}^s {\ell \choose s'}$ possible pebble configurations and $2^\ell$ possible ghost configurations, so the total number of vertices is $|V| = 2^\ell \cdot \sum_{s'=0}^s {\ell \choose s'}$.
Furthermore the graph has rather high degree: for a vertex whose configuration has $s'$ pebbles, there are $4^{s'}$ valid parallel pebbling moves (on each pebble, apply a pebble operation or not, and afterward ghost that pebble or not).
Fortunately we may dramatically reduce the search space by making the following observations, which can be combined into an admissible heuristic for $A*$ search.

First, as we discuss in Section~\ref{sec:parallelspookypebbling}, an optimal pebbling provably always begins with an initial ``blast'' phase where we pebble directly across to position $\ell$, leaving some ``marker'' pebbles along the way.
Therefore we can reduce our search to finding the shortest path from \textit{any} vertex $V_m$ that has a pebble at position $\ell$, to $V_f$, since given a particular $V_m$ the steps from $V_i$ to $V_m$ are entirely specified by this rule.

For the remaining search (corresponding to the ``unblast'' phase), we can make the following observation. 
$\mathsf{Ghost}(i)$ and $\unpeb(i)$ have the same effect of removing a pebble, but $\mathsf{Ghost}(i)$ is essentially free.
Also, if there are any pebbles with index higher than $i$ that need to be removed, then position $i$ will be pebbled again at some point in the future, at which time any ghosts can be removed.
Therefore, the configuration of ghosts in any optimal pebbling is entirely determined by the pebble configuration: there will be ghosts on all positions extending to the highest-indexed position that has a pebble (excluding position $\ell$ itself), and no ghosts on positions with indices higher than that.
This fully removes the factor of $2^\ell$ from the number of vertices we must explore.

Next we observe that $\mathsf{Ghost}(i)$ should only ever be applied immediately after $\pebop(i+1)$ (i.e. after position $i$ is used as input).
The proof is straightforward: given a pebbling scheme in which this is not true, simply move $\mathsf{Ghost}(i)$ back in time to the time slot immediately after $\pebop(i+1)$.
This can never make a pebble game less optimal; if anything it opens opportunity for something more optimal, because it frees up a pebble earlier.
This rule reduces the in-degree of vertices to $2^{s'}$ from $4^{s'}$, where as above $s'$ is the number of pebbles in the vertex's configuration.

Putting these observations together, the key is to start our search from $V_f$ and work \emph{backwards} until we encounter the first vertex $V_m$ in which position $\ell$ is pebbled.
Given that, there is one final addition to the A* heuristic which gives a dramatic further speedup.
We observe that a configuration whose rightmost pebble (exluding $\ell$) is at position $i_\mathrm{max}$ \emph{must} have distance at least $\ell-i_\mathrm{max}$ from any vertex $V_m$, since the rightmost pebble position can change by at most 1 with each parallel pebble step.
Including this fact in the A* heuristic leads to a search that is equivalent to looking for solutions of depth $2\ell$ first, then solutions of depth $2\ell + 1$, etc.
The number of valid sequences grows rapidly with how far one is from a depth $2\ell$ solution, so this strategy ends up exploring a small number of promising solutions first, before expanding into more numerous possibilities.
Furthermore it leads to a nice balance as $s$ is varied: with increased $s$, the number of vertices grows quickly, but the optimal depth also becomes closer to $2\ell$.
In practice we find that this makes it possible to successfully find optimal solutions for a wide range of $s$.

\vspace{1em}

Finally, we note that when searching for optimal pebblings for the practical application of Regev's factoring algorithm, we include the space cost of the temporary ancillas needed to hold the value $t$ (see Section~\ref{sec:factoringcosts}) in the optimization.
This space is only required in $\pebop(i)$ and $\unpeb(i)$ for which $i \equiv 1 \pmod{w+1}$ leading to nontrivial behavior in the pebbling games, as these steps can only be performed when there is sufficient space to spare.

\section{Additional Background on Lattice Reduction}

\subsection{Heuristic Justification for Assumption~\ref{assumption:roothermitetosbbp}}\label{sec:assumptionjustification}

Let us now reason heuristically about how we can use a lattice basis reduction algorithm to solve $\sbbp$ in dimension $r$.
Suppose that it achieves root-Hermite factor $\delta_0$ and outputs a size-reduced\footnote{See~\cite[Section 2.2]{gn08} for a definition of size-reduced bases.} basis $\vecb_1, \ldots, \vecb_r$ with Gram-Schmidt orthogonalization $\vecb^*_1, \ldots, \vecb^*_r$.
Then, as derived in~\cite[Section 3.3.1]{ekeragartnercomparison}, under the Geometric Series Assumption we may approximate $\norm{\vecb^*_i} \approx \norm{\vecb^*_1}/\delta_0^{2(i-1)}$.
Following the blueprint from~\cite[Claim 5.1]{Regev23}, to solve $\sbbp$ we will let $\ell \geq 0$ be minimal such that $\norm{\vecb^*_{\ell+1}} > \delta_0^{2(r-1)} T$ and output $\vecb_1, \ldots, \vecb_\ell$.
Heuristically, we expect that $\norm{\vecb^*_m} > T$ for all $m \geq \ell+1$, which is enough for correctness.
To bound the norm of each $\vecb_i$ for $i \in [\ell]$, note that we have:
$$\norm{\vecb_i}^2 \leq \sum_{j = 1}^i \frac{1}{4} \norm{\vecb_j^*}^2 \leq \frac{r}{4} \cdot \delta_0^{4(r-1)} T^2 \Rightarrow \norm{\vecb_i} \leq \frac{\sqrt{r}}{2} \cdot \delta_0^{2r} \cdot T.$$
(The factor of $1/4$ above comes from the assumption that the basis is size-reduced.)

\subsection{Root-Hermite Factors and Costs for BKZ Lattice Reduction}

\begin{table}[H]
\begin{center}
    \begin{tabular}{|c|c|c|}
    \hline
    Block size & Root-Hermite factor & Estimated CPU clock cycles for dimension 600 \\
    \hline\hline
    60 & 1.01145310214785 & $2^{52}$\\
    \hline
    120 & 1.00843474281592 & $2^{64}$\\
    \hline
    160 & 1.00718344897388 & $2^{76}$\\
    \hline
    200 & 1.00628260691082 & $2^{88}$\\
    \hline
    \end{tabular}
    \caption{Estimated root-Hermite factors and costs (in CPU clock cycles) for BKZ-$\beta$ for various block sizes $\beta$, obtained using Albrecht's estimator~\cite{DBLP:journals/jmc/AlbrechtPS15}. Costs are estimated based on~\cite{DBLP:conf/soda/BeckerDGL16}.}
    \label{roothermitetable}
\end{center}
\end{table}

\subsection{Derivation of Equation~\eqref{eqn:logDbound}}\label{sec:logDderivation}
Recall from Equation~\eqref{eqn:dimensionsinequalityprelims} that we require the following.
In the below derivation, we will plug in Assumption~\ref{assumption:roothermitetosbbp} and approximate $C \approx 1$ (where $C$ is as in Conjecture~\ref{conjecture:regevnt}).
\begin{align*}
\alpha \sqrt{m+1}\, 2^{Cn/d} 
  &< \frac{\sqrt{2}\, R}{\sqrt{d}} \cdot \frac{1}{6}\,(4 \cdot 2^n)^{-1/m} \\
  & \leq \frac{\sqrt{2}\, D}{2d} \cdot \frac{1}{6}\,(4 \cdot 2^n)^{-1/m}\\
\Leftrightarrow\; 
\sqrt{m+d}\,\delta_0^{2(m+d)} \sqrt{m+1}\, 2^{n/d}
  &< \frac{\sqrt{2}\, D}{d} \cdot \frac{1}{6}\,(4 \cdot 2^n)^{-1/m}\\
\Leftrightarrow\; 
\log D 
  &> 2(m+d)\log \delta_0 + \frac{n}{d} + \frac{n}{m}\\
  \quad &+ \tfrac{1}{2}\log(m+d) + \tfrac{1}{2}\log(m+1) 
     + \log\!\frac{6d}{\sqrt{2}} + \tfrac{2}{m}.
\end{align*}

\section{Proof of Theorem~\ref{thm:asymptoticwithoptimalpeb}}\label{app:asymptotic_factoring_proof}

\sloppy The operation $\pebop(i')$ (in the notation of Section~\ref{sec:parallelspookypebbling}), which sets $y_{i'} \leftarrow y_{i'-1}^2 \bmod N$ for odd $i'$ and $y_{i'} \leftarrow \left(\prod_j a_{j}^{z_{j,i'/2}}\right) y_{i'-1}$ for even $i'$, can be (in either case) implemented by a single out-of-place quantum multiplication.
(The value $t = \left(\prod_j a_{j}^{z_{j,i'/2}}\right)$ can be computed with negligible cost~\cite{Regev23,rv24}).
Since $\log t \leq \log \left[ \prod_j a_j \right] = o(n)$ this second multiplication will actually be cheaper than a general $n$-bit multiplication; here we count it the same for simplicity.
The unpebble operation $\unpeb(i')$ is equivalent.
Thus the multiplication depth of $\pebop(i')$ and $\unpeb(i')$ are both 1, and the ancillary space needed is $S_\times(n) + o(n)$.

\fussy
By Theorem~\ref{thm:parallelpebblingmain}, a line graph of length $\ell$ can be pebbled in depth $2\ell$ pebble operations using $2.47\log \ell$ pebbles.
Here $\ell = 2\log D$, so the depth will be $4\log D$ $n$-bit multiplications.
In terms of space, the maximum number of pebble operations that can occur in parallel is half of the number of pebbles (since each operation involves two pebbles, the input and the output).
Thus since one pebble corresponds to $n$ qubits, we use space at most $1.24 (2n + S_\times(n) + o(n)) \log \log D$.
\qed

\section{Modifying Fibonacci-Based Regev Circuits to Use Parallelism}\label{sec:parallelfib}

\cite[Theorem 2.4]{cryptoeprint:2024/636} provides a circuit that uses
\begin{equation}\label{eq:sequentialfibgates}
    2 \cdot \left(\frac{3r}{s} + 4 \log s + 7\right) \cdot K
\end{equation}
mod $N$ multiplications and at least
\begin{equation}\label{eq:sequentialfibspace}
    dK\lceil \log(r+1) \rceil + \max(2\log s + 5, 7)n
\end{equation}
qubits.

\paragraph{Straightforward optimizations.}
A closer look reveals that many of the operations in this circuit can be optimized in parallelism and/or space.

\begin{itemize}
    \item In~\cite[Lemma 3.5]{cryptoeprint:2024/636}, we need only use $n$ ancilla qubits (rather than $2n$); for typical multiplication circuits~\cite{kahanamokumeyer2024fast}, there is no need to write down a classical constant in a separate register before multiplying it in.
    \sloppy
    \item In~\cite[Lemma 3.7]{cryptoeprint:2024/636}, the quantum circuit computing $a, a^{-1}, b, b^{-1} \to a, a^{-1}, ab, (ab)^{-1} \pmod{N}$ can be implemented using 2 parallel blocks of 2 multiplications each and $\approx 2n$ ancilla qubits, rather than 3 sequential multiplications and $n$ ancilla qubits.

    This is because we can compute $a, b \to a, b, ab$ and $a^{-1}, b^{-1} \to a^{-1}, b^{-1}, a^{-1}b^{-1}$ in parallel before uncomputing $b$ using $a^{-1}, ab$ and uncomputing $b^{-1}$ using $a, a^{-1}b^{-1}$ in parallel.
    
    \fussy
    \item \cite[Lemma 3.8]{cryptoeprint:2024/636} can be modified in two places: in the first part, we can use a circuit for squaring modulo $N$ without having to write the number down twice.
    In the second part, we can plug in the aforementioned parallelized version of~\cite[Lemma 3.7]{cryptoeprint:2024/636}.
    This gives a circuit with a depth of $2r/s + 2\log s$ mod $N$ multiplications, a total of $4r/s + 4\log s$ total mod $N$ multiplications, and $(2\log s + 2)n$ ancilla qubits.
\end{itemize}

\paragraph{Running steps 3a and 3b in parallel.}
Finally, we can run steps 3a and 3b of~\cite[Algorithm 4.1]{cryptoeprint:2024/636} in parallel.
In more detail, we will start by running all but two stages of mod $N$ multiplications in~\cite[Lemma 3.8]{cryptoeprint:2024/636} as used in step 3a.
We will then run step 3b(i) on its own, then run the last two parallel blocks of mod $N$ multiplications in~\cite[Lemma 3.8]{cryptoeprint:2024/636} in parallel with the two multiplications in~\cite[Lemma 3.5]{cryptoeprint:2024/636} as used in step 3b(ii).
Finally, we will run step 3b(iii) on its own.
In all cases, we will have a depth of $2r/s + 2\log s$ mod $N$ multiplications and a total number of $4r/s + 4\log s + 2$ mod $N$ multiplications.
Accounting for the space is more delicate; for clarity, we step through this in cases below.

\emph{Case 1: $s = 1$.}
In this case, we will initially require $4n+o(n)$ ancilla qubits, out of which $2n$ will be restored.
After $2(r-1)$ parallel blocks of multiplications and applying 3b(i), we will have (omitting $o(n)$ qubits throughout):
\begin{align*}
    &\ket{a} \ket{a^{-1}} \ket{a^{r-1}b} \ket{(a^{r-1}b)^{-1}} \ket{\prod a_i^{r-z_{i, j}}} \ket{0^{3n}} \\
    \rightarrow &\ket{a} \ket{a^{-1}} \ket{a^{r-1}b} \ket{(a^{r-1}b)^{-1}} \ket{\prod a_i^{r-z_{i, j}}} \ket{a^rb} \ket{(a^rb)^{-1}} \ket{\prod a_i^{-z_{i, j}}} \\
    \rightarrow &\ket{a} \ket{a^{-1}} \ket{a^rb} \ket{(a^rb)^{-1}} \ket{\prod a_i^{-z_{i, j}}} \ket{0^{3n}} \\
    \rightarrow &\ket{a} \ket{a^{-1}} \ket{a^rb} \ket{(a^rb)^{-1}} \ket{\prod a_i^{-z_{i, j}}} \ket{\prod a_i^{z_{i, j}}} \ket{0^{2n}}.
\end{align*}

\emph{Case 2: $s = 2$.}
In this case, we will initially require $6n+o(n)$ ancilla qubits, out of which $4n$ will be restored.
After $r$ parallel blocks of multiplications and applying 3b(i), we will be halfway through the last application of~\cite[Lemma 3.7]{cryptoeprint:2024/636} within step 3a.
We hence have:
\begin{align*}
    &\ket{a} \ket{a^{-1}} \ket{a^2} \ket{a^{-2}} \ket{a^{r-2}b} \ket{(a^{r-2}b)^{-1}} \ket{a^rb} \ket{(a^rb)^{-1}} \ket{\prod a_i^{r-z_{i, j}}} \ket{0^{n}} \\
    \rightarrow &\ket{a} \ket{a^{-1}} \ket{a^2} \ket{a^{-2}} \ket{a^rb} \ket{(a^rb)^{-1}} \ket{\prod a_i^{r-z_{i, j}}} \ket{\prod a_i^{-z_{i, j}}} \ket{0^{2n}} \\
    \rightarrow &\ket{a} \ket{a^{-1}} \ket{a^rb} \ket{(a^rb)^{-1}} \ket{\prod a_i^{-z_{i, j}}} \ket{0^{5n}} \\
    \rightarrow &\ket{a} \ket{a^{-1}} \ket{a^rb} \ket{(a^rb)^{-1}} \ket{\prod a_i^{-z_{i, j}}} \ket{\prod a_i^{z_{i, j}}} \ket{0^{4n}}.
\end{align*}

\emph{Case 3: $s \geq 4$.}
In this case, we require $(2\log s+2)n+o(n)$ clean ancilla qubits, out of which $(2\log s)n + o(n)$ will be restored.
After $2r/s + 2\log s - 2$ parallel blocks of multiplications and applying 3b(i), we will have:
\begin{align*}
    &\ket{a}\ket{a^{-1}} \ket{a^2} \ket{a^{-2}} \ket{a^4} \ket{a^{-4}} \ket{a^rb} \ket{(a^rb)^{-1}} \ket{\prod a_i^{r-z_{i, j}}} \ket{0^{(2\log s- 3)n}} \\
    \rightarrow &\ket{a}\ket{a^{-1}} \ket{a^2} \ket{a^{-2}} \ket{a^rb} \ket{(a^rb)^{-1}} \ket{\prod a_i^{r-z_{i, j}}} \ket{\prod a_i^{-z_{i, j}}} \ket{0^{(2\log s - 2)n}} \\
    \rightarrow &\ket{a}\ket{a^{-1}} \ket{a^rb} \ket{(a^rb)^{-1}} \ket{\prod a_i^{-z_{i, j}}} \ket{0^{(2\log s + 1)n}} \\
    \rightarrow &\ket{a}\ket{a^{-1}} \ket{a^rb} \ket{(a^rb)^{-1}} \ket{\prod a_i^{-z_{i, j}}} \ket{\prod a_i^{z_{i, j}}} \ket{0^{(2\log s)n}}.
\end{align*}

\paragraph{Putting everything together, assuming $s \in \{1, 2\}$.}
As in all our results in practice as well as those of~\cite{ekeragartnercomparison}, we typically always have $s = 1, 2$. In this case, running 3a and 3b in parallel initially requires $2\log s + 4$ clean ancilla qubits.
We claim that the total space usage is then $dK\lceil \log(r+1) \rceil + (2\log s + 8)n$ qubits:
\begin{itemize}
    \item At the beginning of step 3, this will leave us with $(2\log s + 4)n$ clean ancilla qubits.
    \item We can run steps 3a and 3b in parallel as per the above. At this point, we will be left with $(2 \log s + 2)n$ clean ancilla qubits.
    \item For step 3c, we only need $2n$ clean ancilla qubits, which we have and will restore.
    \item For step 3d, we only need $n$ ancilla qubits, which we have and will restore.
\end{itemize}

\section{Parallelizing Shor's Algorithm}
\label{app:parallel_shor}

In this appendix we discuss the parallelization of the modular exponentiation in Shor's algorithm~\cite{shor97}.
The structure does not meaningfully change when leveraging e.g. the innovations of~\cite{DBLP:conf/pqcrypto/EkeraH17}, so for simplicity this section is written in terms of the construction in Shor's original paper.
Note that here we do not consider the use of recent innovations using Residue Number Systems to implement quantum factoring~\cite{DBLP:journals/iacr/ChevignardFS24, gidney2025factor2048bitrsa}, whose depth cannot easily be measured in number of $n$-bit multiplications.

Consider the operation $\ket{x} \ket{0} \to \ket{x}\ket{a^x \bmod N}$, where $a$ and $N$ are classical constants.
This is typically implemented via the decomposition
\begin{equation}
	a^x \bmod N = \prod_i c_i^{x_i} \bmod N
\end{equation}
where $c_i = a^{2^i} \bmod N$.
As a quantum circuit, this is implemented by starting with the state $\ket{x}\ket{1}$ and then for each $i$, performing an in-place modular multiplication of the second register by $c_i^{x_i}$.
This in-place multiplication is built out of two out-of-place modular multiply-adds using an ancilla register. 
The first is a multiplication $\ket{z}\ket{0} \to \ket{z}\ket{z c_i^{x_i} \bmod N}$ which yields the desired result in the ancilla register, and the second uses the modular inverse of $c_i$ (precomputed classically) to uncompute the input: $\ket{z}\ket{z c_i^{x_i} \bmod N} \to \ket{z - (c_i^{-1})^{x_i} \cdot z c_i^{x_i} \bmod N}\ket{z c_i^{x_i} \bmod N} = \ket{0}\ket{z c_i^{x_i} \bmod N}$.
A swap may be performed to move this value into the original register, but it is not necessary in practice as the first register can now be redefined as the ancilla register for the next iteration.

A particularly nice feature of this construction is that a single qubit can be recycled for the $\ket{x}$ register (see e.g.~\cite{beauregard}).
The total qubit count of the factoring algorithm is thus roughly $2n$ (plus the workspace required to do an out-of-place multiplication): one qubit for $\ket{x}$, one register for the output, and a second ancilla register.

One may parallelize this process by breaking the bits of $x$ into $p$ chunks $x[i]$ such that $\sum_i x[i] = x$ and then separately computing $a^{x[i]} \bmod N$ in parallel.
This yields the state
\begin{equation}
	\ket{x}\ket{a^{x[0]} \bmod N}\ket{a^{x[1]} \bmod N} \cdots \ket{a^{x[p-1]} \bmod N}
\end{equation}
We may then compute $\ket{a^{x} \bmod N}$ by multiplying these values together, an operation which can be cast as a spooky pebble game on a line graph of length $p-1$ (different geometries for the graph are possible too, since multiplication is associative).

Unfortunately we now run into a bit of a problem: there is no cheap way to uncompute the partial products we produced.
Indeed, it seems that if one follows this plan it is necessary to do the entire partial modular exponentiation that produced these registers in reverse to uncompute them, which introduces an extra factor of two to the number of multiplications required. 
This then increases the space usage, as the need to do the computation in reverse means that qubit recycling can no longer be applied to the $\ket{x}$ register, so an entire register of $n$ qubits is now needed for that.

The good news is that these overheads do not grow with $p$.
If $M$ is the multiplication depth of the original sequential construction, the depth will be $2\lceil M/p \rceil$ multiplications plus the depth of the length $p-1$ spooky pebble game at the end, and the qubit cost will be roughly $(1 + 2p)n$ qubits plus those needed for the final spooky pebble game (here ignoring ancillas required for the out-of-place multiplications).

Setting $p=2$ yields no advantage in depth, because the improvement is canceled out by the extra cost of uncomputation.
For our estimates in this work, we examine $p=3$ and $p=5$.
For $p=3$ the final spooky pebble game is trivial, it uses $2$ registers and has depth $3$.
For $p=5$ the spooky pebble game is not quite trivial, one can get away with only using $3$ extra registers for a multiplication depth of $8$.
Thus the overall costs are:
\begin{itemize}
	\item $p = 3$: depth $2\lceil M/3 \rceil + 3$, qubits $\approx 9n$
	\item $p = 5$: depth $2\lceil M/5 \rceil + 8$, qubits $\approx 14n$
\end{itemize}

\section{Detailed Results for Factoring 2048- and 4096-Bit Integers}

\subsection{Shor's Algorithm}

\begin{table}[H]
\begin{center}
    \begin{tabular}{|c|c|c|c|c|c|c|c|}
    \hline
    \multirow{2}{*}{$n$} & \multirow{2}{*}{Parallelism} & \multirow{2}{*}{Windowing} & \multirow{2}{*}{N. runs} & \multicolumn{3}{c|}{Mod $N$ multiplications} & \multirow{2}{*}{Qubits$/n$} \\
    \cline{5-7}
    & & & & Depth & Per run & Overall & \\
    \hline\hline
    \multirow{12}{*}{2048} & \multirow{4}{*}{None} & \multirow{2}{*}{No} & 1 & 6018 & 6018 & 6018 & \multirow{4}{*}{$\approx 2$} \\
    \cline{4-7}
    & & & 20 & 2290 & 2290 & 45800 & \\
    \cline{3-7}
    & & \multirow{2}{*}{$w = 10$} & 1 & 602 & 602 & 602 & \\
    \cline{4-7}
    & & & 20 & 230 & 230 & 4600 & \\
    \cline{2-8}
    & \multirow{4}{*}{$p=3$} & \multirow{2}{*}{No} & 1 & 4015 &12039 & 12039 & \multirow{4}{*}{$\approx 9$} \\
    \cline{4-7}
    & & & 20 & 1531 & 4583 & 91660 & \\
    \cline{3-7}
    & & \multirow{2}{*}{$w = 10$} & 1 & 405 & 1207 & 1207 & \\
    \cline{4-7}
    & & & 20 & 157 & 463 & 9260 & \\
    \cline{2-8}
    & \multirow{4}{*}{$p=5$} & \multirow{2}{*}{No} & 1 & 2416 & 12044 & 12044 & \multirow{4}{*}{$\approx 14$} \\
	\cline{4-7}
	& & & 20 & 924 & 4588 & 91760 & \\
	\cline{3-7}
	& & \multirow{2}{*}{$w = 10$} & 1 & 250 & 1212 & 1212 & \\
	\cline{4-7}
	& & & 20 & 100 & 468 &9360 & \\
	\hline
	\multirow{12}{*}{4096} & \multirow{4}{*}{None} & \multirow{2}{*}{No} & 1 & 12162 & 12162 & 12162 &\multirow{4}{*}{$\approx 2$} \\
	\cline{4-7}
	& & & 27 & 4438 & 4438 & 119826 & \\
	\cline{3-7}
	& & \multirow{2}{*}{$w = 10$} & 1 & 1218 & 1218 & 1218 & \\
	\cline{4-7}
	& & & 27 & 444 & 444 & 11988 & \\
	\cline{2-8}
	& \multirow{4}{*}{$p=3$} & \multirow{2}{*}{No} & 1 & 8111 & 24327 & 24327 & \multirow{4}{*}{$\approx 9$} \\
	\cline{4-7}
	& & & 27 & 2963 & 8879 & 239733 & \\
	\cline{3-7}
	& & \multirow{2}{*}{$w = 10$} & 1 & 815 & 2439 & 2439 & \\
	\cline{4-7}
	& & & 27 & 299 & 891 & 24057 & \\
	\cline{2-8}
	& \multirow{4}{*}{$p=5$} & \multirow{2}{*}{No} & 1 & 4874 & 24332 & 24332 & \multirow{4}{*}{$\approx 14$} \\
	\cline{4-7}
	& & & 27 & 1784 & 8884 & 239868 & \\
	\cline{3-7}
	& & \multirow{2}{*}{$w = 10$} & 1 & 496 & 2444 & 2444 & \\
	\cline{4-7}
	& & & 27 & 186 & 896 & 24192 & \\
	\hline
    \end{tabular}
    \caption{Costs for Shor's algorithm~\cite{shor97} for $n = 2048, 4096$ with optimizations by~\cite{DBLP:conf/pqcrypto/EkeraH17,DBLP:journals/dcc/Ekera20,DBLP:journals/corr/abs-2309-01754} as evaluated and reported by~\cite{ekeragartnercomparison}, plus parallelism as described in Appendix~\ref{app:parallel_shor}. We differentiate along three axes: whether parallelism is used, whether windowing is used~\cite[Table 2]{ekeragartnercomparison} or not~\cite[Table 1]{ekeragartnercomparison}, and whether we use one or many runs of the quantum circuit. As noted by~\cite{ekeragartnercomparison}, it is possible that Shor would benefit from even more aggressive windowing than $w = 10$, but here we adhere to the parameter settings used by their work.}
    \label{shortable}
\end{center}
\end{table}

\subsection{Regev's Factoring Algorithm with Fibonacci Arithmetic}\label{app:regevfib}

First, we state explicit cost estimates from essentially using~\cite[Algorithm 4.1]{cryptoeprint:2024/636} as in~\cite{ekeragartnercomparison}, but with some key modifications to benefit from parallelism which are discussed in Appendix~\ref{sec:parallelfib}.

Let $s \leq r$ be positive integers such that $s$ is a power of 2 and $s$ divides $r$.
These are parameters that the algorithm designer is free to choose.
Then define the sequence $\{G_k\}$ by $G_0 = 0, G_1 = 1$, and $G_k = rG_{k-1} + G_{k-2}$ for $k \geq 2$.
Let $K$ be maximal such that $G_K \leq D$.
We show that the necessary computation can then done with a depth of
\begin{equation}\label{eq:parallelfibdepth}
    2 \cdot \left(\frac{2r}{s} + 2 \log s + 4\right) \cdot K
\end{equation}
mod $N$ multiplications, a total of
\begin{equation}\label{eq:parallelfibgates}
    2 \cdot \left(\frac{4r}{s} + 4 \log s + 8\right) \cdot K
\end{equation}
mod $N$ multiplications, and --- assuming\footnote{We also treat the case of $s \geq 4$ in Appendix~\ref{sec:parallelfib}, but we only need these cases for our results.} $s \in \{1, 2\}$ --- a qubit count of
\begin{equation}\label{eq:parallelfibspace}
    \geq dK\lceil \log(r+1) \rceil + (2\log s + 8)n + o(n),
\end{equation}
which is a fairly aggressive lower bound because we are including the cost of storing a product tree computing $\prod_{j = 1}^d a_j^{z_{j, i}}$ in the $o(n)$ term.

With the above modifications in mind, our experimental methodology for evaluating Fibonacci-based approaches to Regev's circuit is straightforward.
For any $n, \beta, r, s$ (such that $s$ is a power of 2 dividing $r$), we select $d, m$ by iterating over all possible values subject to constraints analogous to those stated in ``Setting parameters in our algorithms'' (the constraint on $d$ will instead be $\prod_{i = 1}^d p_i^r \leq 2^n$), and minimizing the required value of $\log D$ as dictated by Equation~\eqref{eqn:logDbound}.
We then set $K$ to be maximal such that $G_K \leq D$ and calculate and report costs according to Equations~\eqref{eq:parallelfibdepth} and~\eqref{eq:parallelfibgates}.

These comparison points are summarized in Table~\ref{fibtable}.
For the rows corresponding to $(r, s) = (1, 1)$, we report costs exactly as specified above.
For the other rows, we search over values of $r, s$ and report the results achieving minimal multiplication depth.
In all our settings, it happens that these values are $(r, s) = (4, 2)$.

We observe that --- perhaps surprisingly --- the concrete gain from increasing BKZ block size is somewhat marginal. This will be apparent in our later spooky pebbling results as well.

\begin{table}
\begin{center}
    \begin{tabular}{|c|c|c|c|c|c|c|c|c|c|c|}
    \hline
    \multirow{2}{*}{$n$} & \multirow{2}{*}{BKZ-$\beta$} & \multirow{2}{*}{$r, s$} & \multirow{2}{*}{$d$} & \multirow{2}{*}{$m$} & \multirow{2}{*}{$\log D$} & \multirow{2}{*}{$K$} & \multicolumn{3}{c|}{Mod $N$ multiplications} & \multirow{2}{*}{Qubits$/n$} \\
    \cline{8-10}
    & & & & & & & Depth & Per run & Overall & \\
    \hline
    \hline
    \multirow{8}{*}{2048} & \multirow{2}{*}{60} & $1, 1$ & 165 & 227 & 52 & 75 & 900 & 1800 & 408600 & $\geq 14.0$  \\
    \cline{3-11}
    & & 4, 2 & 74 & 210 & 63 & 31 & 620 & 1240 & 260400 & $\geq 13.3$ \\
    \cline{2-11}
    & \multirow{2}{*}{120} & $1, 1$ & 177 & 260 & 48 & 69 & 828 & 1656 & 430560 & $\geq 13.9$ \\
    \cline{3-11}
    & & $4, 2$ & 75 & 228 & 60 & 29 & 580 & 1160 & 264480 & $\geq 13.1$\\
    \cline{2-11}
    & \multirow{2}{*}{160} & $1, 1$ & 189 & 277 & 46 & 66 & 792 & 1584 & 438768 & $\geq 14.0$ \\
    \cline{3-11}
    & & $4, 2$ & 74 & 253 & 59 & 29 & 580 & 1160 & 293480 & $\geq 13.1$\\
    \cline{2-11}
    & \multirow{2}{*}{200} & $1, 1$ & 222 & 303 & 44 & 64 & 768 & 1536 & 465408 & $\geq 14.9$ \\
    \cline{3-11}
    & & $4, 2$ & 74 & 278 & 58 & 28 & 560 & 1120 & 311360 & $\geq 13.0$\\
    \hline\hline
    \multirow{8}{*}{4096} & \multirow{2}{*}{60} & $1, 1$ & 290 & 333 & 66 & 95 & 1140 & 2280 & 759240 & $\geq 14.7$ \\
    \cline{3-11}
    & & $4, 2$ & 129 & 315 & 77 & 37 & 740 & 1480 & 466200 & $\geq 13.4$\\
    \cline{2-11}
    & \multirow{2}{*}{120} & $1, 1$ & 317 & 381 & 60 & 87 & 1044 & 2088 & 795528 & $\geq 14.7$ \\
    \cline{3-11}
    & & $4, 2$ & 128 & 363 & 73 & 35 & 700 & 1400 & 508200 & $\geq 13.2$\\
    \cline{2-11}
    & \multirow{2}{*}{160} & $1, 1$ & 358 & 411 & 57 & 82 & 984 & 1968 & 808848 & $\geq 15.1$ \\
    \cline{3-11}
    & & $4, 2$ & 129 & 367 & 71 & 34 & 680 & 1360 & 499120 & $\geq 13.2$\\
    \cline{2-11}
    & \multirow{2}{*}{200} & $1, 1$ & 354 & 434 & 55 & 79 & 948 & 1896 & 822864 & $\geq 14.8$ \\
    \cline{3-11}
    & & $4, 2$ & 131 & 427 & 69 & 34 & 680 & 1360 & 580720 & $\geq 13.2$\\
    \hline
    \end{tabular}
    \caption{Costs for Regev's algorithm for $n = 2048, 4096$ based on the Fibonacci-style arithmetic in~\cite{rv24,cryptoeprint:2024/636}. We differentiate based on BKZ block size $\beta$ (which corresponds to classical postprocessing power) and whether we set $(r, s) = (1, 1)$ or allow $r > 1$. Our numbers deviate from those of~\cite{ekeragartnercomparison} for reasons discussed in Appendix~\ref{app:regevfib}. We additionally report \textbf{lower bounds} on the qubit usage, ignoring --- in addition to qubit costs described at the beginning of Section~\ref{sec:facsetup} --- that we consistently neglect the qubits required to hold the values $\prod_{j = 1}^d a_j^{z_{j, i}}$. We do not ignore these costs in Tables~\ref{table2048} and~\ref{table4096}.}
    \label{fibtable}
\end{center}
\end{table}

\begin{table}
\begin{center}
    \begin{tabular}{|c|c|c|c|c|c|c|c|c|}
    \hline
    \multirow{2}{*}{BKZ-$\beta$} & \multirow{2}{*}{$d$} & \multirow{2}{*}{$m$} & \multirow{2}{*}{$\log D$} & \multirow{2}{*}{$s$} & \multicolumn{3}{c|}{Mod $N$ multiplications} & \multirow{2}{*}{Qubits$/n$} \\
    \cline{6-8}
    & & & & & Depth & Per run & Overall & \\
    \hline
    \hline
\multirow{3}{*}{60} & \multirow{3}{*}{90} & \multirow{3}{*}{203} & \multirow{3}{*}{59} & 5 & 465 & 500 & 101500 & 7.6 \\
\cline{5-9}
&  &  &  & 7 & 263 & 370 & 75110 & 9.6 \\
\cline{5-9}
&  &  &  & 12 & 175 & 295 & 59885 & 14.6 \\
\hline
\hline
\multirow{3}{*}{120} & \multirow{3}{*}{90} & \multirow{3}{*}{242} & \multirow{3}{*}{56} & 5 & 418 & 451 & 109142 & 7.5 \\
\cline{5-9}
&  &  &  & 7 & 240 & 341 & 82522 & 9.5 \\
\cline{5-9}
&  &  &  & 12 & 163 & 266 & 64372 & 14.5 \\
\hline
\hline
\multirow{3}{*}{160} & \multirow{3}{*}{94} & \multirow{3}{*}{242} & \multirow{3}{*}{54} & 5 & 396 & 429 & 103818 & 7.5 \\
\cline{5-9}
&  &  &  & 7 & 229 & 326 & 78892 & 9.5 \\
\cline{5-9}
&  &  &  & 12 & 157 & 253 & 61226 & 14.5 \\
\hline
\hline
\multirow{3}{*}{200} & \multirow{3}{*}{94} & \multirow{3}{*}{255} & \multirow{3}{*}{53} & 5 & 396 & 429 & 109395 & 7.4 \\
\cline{5-9}
&  &  &  & 7 & 229 & 326 & 83130 & 9.4 \\
\cline{5-9}
&  &  &  & 12 & 157 & 253 & 64515 & 14.4 \\
\hline
    \end{tabular}
    \caption{Costs for Regev's algorithm with parallel spooky pebbling for $n = 2048$. For four different BKZ block sizes $\beta$ (which corresponds to classical postprocessing power), we present two sets of results corresponding to a tradeoff between qubit count and depth. For all presented results, the window size $w=2$. Qubit counts here ignore the minor overheads described at the beginning of Section~\ref{sec:facsetup}. However, unlike Table~\ref{fibtable}, this table \emph{does} account for the qubits needed to hold $\prod_{j = 1}^d a_j^{z_{j, i}}$.}
    \label{table2048}
\end{center}
\end{table}

\begin{table}
\begin{center}
    \begin{tabular}{|c|c|c|c|c|c|c|c|c|}
    \hline
    \multirow{2}{*}{BKZ-$\beta$} & \multirow{2}{*}{$d$} & \multirow{2}{*}{$m$} & \multirow{2}{*}{$\log D$} & \multirow{2}{*}{$s$} & \multicolumn{3}{c|}{Mod $N$ multiplications} & \multirow{2}{*}{Qubits$/n$} \\
    \cline{6-8}
    & & & & & Depth & Per run & Overall & \\
    \hline
    \hline
\multirow{3}{*}{60} & \multirow{3}{*}{162} & \multirow{3}{*}{314} & \multirow{3}{*}{72} & 5 & 610 & 652 & 204728 & 7.8 \\
\cline{5-9}
&  &  &  & 7 & 334 & 475 & 149150 & 9.8 \\
\cline{5-9}
&  &  &  & 8 & 287 & 444 & 139416 & 10.8 \\
\hline
\hline
\multirow{3}{*}{120} & \multirow{3}{*}{166} & \multirow{3}{*}{371} & \multirow{3}{*}{67} & 5 & 561 & 600 & 222600 & 7.7 \\
\cline{5-9}
&  &  &  & 7 & 310 & 446 & 165466 & 9.7 \\
\cline{5-9}
&  &  &  & 12 & 200 & 381 & 141351 & 14.7 \\
\hline
\hline
\multirow{3}{*}{160} & \multirow{3}{*}{166} & \multirow{3}{*}{388} & \multirow{3}{*}{65} & 5 & 537 & 575 & 223100 & 7.6 \\
\cline{5-9}
&  &  &  & 7 & 298 & 432 & 167616 & 9.6 \\
\cline{5-9}
&  &  &  & 12 & 193 & 345 & 133860 & 14.6 \\
\hline
\hline
\multirow{3}{*}{200} & \multirow{3}{*}{162} & \multirow{3}{*}{401} & \multirow{3}{*}{64} & 5 & 513 & 550 & 220550 & 7.5 \\
\cline{5-9}
&  &  &  & 7 & 286 & 419 & 168019 & 9.5 \\
\cline{5-9}
&  &  &  & 12 & 187 & 326 & 130726 & 14.5 \\
\hline
\end{tabular}
    \caption{Costs for Regev's algorithm with parallel spooky pebbling for $n = 4096$. For four different BKZ block sizes $\beta$ (which corresponds to classical postprocessing power), we present two sets of results corresponding to a tradeoff between qubit count and depth. For all presented results, the window size $w=2$. Qubit counts here ignore the minor overheads described at the beginning of Section~\ref{sec:facsetup}. However, unlike Table~\ref{fibtable}, this table \emph{does} account for the qubits needed to hold $\prod_{j = 1}^d a_j^{z_{j, i}}$.}
    \label{table4096}
\end{center}
\end{table}

\end{document}